\documentclass[lettersize,journal]{IEEEtran}
\pdfoutput=1
\usepackage{fullpage}
\usepackage[bookmarks]{hyperref}
\usepackage{amssymb}
\usepackage{amsmath}
\usepackage{amsthm}
\usepackage{tcolorbox}
\usepackage{graphicx,color,colordvi}
\usepackage{bbm}
\usepackage{cleveref}
\usepackage{stmaryrd}
\usepackage[utf8]{inputenc}
\usepackage{fullpage}
\usepackage[blocks]{authblk}
\usepackage{dsfont}
\usepackage{mathtools}
\usepackage{authblk}

\usepackage{centernot}

\usepackage{pgfplots}

\def\openone{\leavevmode\hbox{\small1\kern-3.8pt\normalsize1}}

\def\DD{\mathbb{D}}

\def\11{\mathbb{I}}

\newtheorem{definition}{Definition}[section]
\newtheorem{proposition}[definition]{Proposition}
\newtheorem{lemma}[definition]{Lemma}

\newtheorem{theorem}[definition]{Theorem}
\newtheorem{corollary}[definition]{Corollary}

\def\eps{\varepsilon}

\newcommand{\tr}{\mathop{\rm Tr}\nolimits}

\usepackage{pst-node}
\usepackage{tikz-cd}

\newcommand{\cA}{{\cal A}}

\newcommand{\cC}{{\cal C}}
\newcommand{\cD}{{\cal D}}

\newcommand{\cN}{{\cal N}}
\newcommand{\cT}{{\cal T}}
\newcommand{\cH}{{\cal H}}

\newcommand{\cP}{\mathcal{P}}

\newcommand{\cR}{{\cal R}}

\newcommand{\cM}{{\mathcal{M}}}

\newcommand{\Id}{{\mathds{1}}}

\newcommand{\cSe}[1]{\mathcal{S}_{=}(\cH_{#1})}
\newcommand{\cSs}[1]{\mathcal{S}_{\leq}(\cH_{#1})}
\newcommand{\cPo}[1]{\mathcal{P}_{}(\cH_{#1})}

\newcommand{\Renyi}{R{\'e}nyi~}

\usepackage{graphicx}
\usepackage{setspace}
\usepackage{verbatim}
\usepackage{subfig}

\numberwithin{equation}{section}

\DeclareRobustCommand\openone{\leavevmode\hbox{\small1\normalsize\kern-.33em1}}
\newcommand{\identity}{I}

\newcommand{\be}{\begin{equation}}
	\newcommand{\ee}{\end{equation}}
\newcommand{\bea}{\begin{eqnarray}}
	\newcommand{\eea}{\end{eqnarray}}
\newcommand{\beas}{\begin{eqnarray*}}
	\newcommand{\eeas}{\end{eqnarray*}}

\title{Quantum Differential Privacy: \\ An Information Theory Perspective}

\begin{document}

\author{Christoph Hirche, Cambyse Rouz\'{e}, Daniel Stilck Fran\c ca
\thanks{C. Hirche is with the Zentrum Mathematik, Technical University of Munich, 85748 Garching, Germany 
 and the Centre for Quantum Technologies, National University of Singapore, Singapore}
\thanks{C. Rouz\'{e} is with the Munich Center for Quantum Science and Technology (MCQST), M\"unchen, Germany}
\thanks{D. S. Fran\c ca is with QMATH, Department of Mathematical Sciences, University of Copenhagen, Universitetsparken 5, 2100 Copenhagen, Denmark}
}

\date{}

\maketitle

\begin{abstract}
Differential privacy has been an exceptionally successful concept when it comes to providing provable security guarantees for classical computations. More recently, the concept was generalized to quantum computations. While classical computations are essentially noiseless and differential privacy is often achieved by artificially adding noise, near-term quantum computers are inherently noisy and it was observed that this leads to natural differential privacy as a feature. 

In this work we discuss quantum differential privacy in an information theoretic framework by casting it as a quantum divergence. A main advantage of this approach is that differential privacy becomes a property solely based on the output states of the computation, without the need to check it for every measurement. This leads to simpler proofs and generalized statements of its properties as well as several new bounds for both, general and specific, noise models. In particular, these include common representations of quantum circuits and quantum machine learning concepts. Here, we focus on the difference in the amount of noise required to achieve certain levels of differential privacy versus the amount that would make any computation useless. Finally, we also generalize the classical concepts of local differential privacy, \Renyi differential privacy and the hypothesis testing interpretation to the quantum setting, providing several new properties and insights. 
\end{abstract}

\section{Introduction}

\IEEEPARstart{P}{rocessing} data in some form is the core concept of most computational tasks. Nowadays, large data sets are being collected and processed for a variety of tasks ranging from medical studies to machine learning applications. With the accumulation of information always also come security concerns. If one presents the results of a study, will that allow the audience to conclude on the health of a particular individual? Social media companies constantly process our data using machine learning for advertisement purposes. How much do the results reveal about the underlying data set? 

Differential privacy~\cite{dwork2006calibrating, dwork2006differential} is a concept introduced in the classical computation setting to address these concerns. Vaguely speaking, differential privacy guarantees that the probability of an algorithm giving a certain outcome is roughly the same for any sufficiently similar input. This implies that good differential privacy makes it difficult for an observer to make precise statements about the data used. 

With the growing interest in large-scale quantum computations and quantum machine learning applications, naturally, also security requirements are desired there. To that end, the concept of quantum differential privacy was introduced in~\cite{zhou2017differential} in the setting of quantum computing. That work was the starting point of a series of results connecting quantum differential privacy to gentle measurements~\cite{aaronson2019gentle}, distributed quantum computing~\cite{li2021quantum} and quantum machine learning~\cite{arunachalam2021private,senekane2017privacy,watkins2021quantum,du2021quantum,kashefi1,kashefi2,min-hsiu-regression}. Other authors also studied differential privacy in the classical to quantum regime~\cite{9174484}.

In quantum differential privacy, we are interested in the properties of a quantum algorithm $\cA$ represented by a general quantum channel. The input data is a quantum state $\rho$ and someone observing the output should not be able to determine whether the input state was indeed $\rho$ or a similar state $\sigma$. Similarity is usually defined via neighbouring states, denoted $\rho\sim\sigma$, according to some rule. 
Different notions of similarity have been proposed in the literature. For instance, a small trace distance~\cite{zhou2017differential} or convertibility by a local quantum channel~\cite{aaronson2019gentle}. For most of our discussion we will keep the definition unspecified and only sometimes fix it for examples. Now, a quantum algorithm $\cA$ is $(\epsilon,\delta)$-differentially private, if for all measurements $M$ and all neighbouring quantum states $\rho\sim\sigma$ we have 
\begin{align}
\tr( M \cA(\rho)) \leq e^\epsilon \tr( M \cA(\sigma)) + \delta, \label{Eq:DP-def-in}
\end{align}
see Definition~\ref{Def:QDP} for full details. In the classical setting, one way to achieve differential privacy is by taking an algorithm and adding noise to the output to obscure the input~\cite{Dwork2013}. This idea was transferred to the quantum setting in~\cite{zhou2017differential}, showing that concatenating an algorithm with sufficiently strong noise in form of a different quantum channel makes the algorithm differentially private. This was further explored in~\cite{du2021quantum} for layered algorithms that are affected by noise at every step, a model that describes many common scenarios such as quantum circuits and quantum machine learning algorithms implemented on noisy quantum computers. An interesting proposition following from~\cite{du2021quantum} is that the noise present in near-term quantum devices, while presenting a computational difficulty, induces an inherent advantage by naturally making the computation differentially private. On the other hand it is of course also well known that such noise can make it impossible to run long computations, somewhat limiting the computational usefulness of near-term devices. One of the main goals of this work is to give upper bounds on the depth required to reach differential privacy and contrasting it to the computational limitations imposed by the noise. 
As we will see, differential privacy can be reached at a significantly shorter depth than that at which computationally prohibitive noise occurs.

To this end, we introduce an information theoretic approach to quantum differential privacy that allows us to cast the requirement in terms of a quantum divergence. This in turn lets us make several observations that are solely based on the properties of the divergence, giving a fruitful new approach to discussing quantum differential privacy. The divergence of choice is the quantum hockey-stick divergence introduced in~\cite{sharma2012strong} and the general framework follows classical work presented among others in~\cite{asoodeh2020privacy,asoodeh2021local} using tools such as contraction coefficients to bound differential privacy in iterative algorithms.  

The outline and main contributions of this work are as follows. 
\begin{itemize}
    \item In Section~\ref{Sec:HS-divergence} we revisit the quantum hockey-stick divergence and show several new properties. In particular we consider its associated contraction coefficient and give a simple expression for it that significantly reduces its computational complexity. These results lay the foundation for what follows, but should also be of independent interest. 
    \item In Section~\ref{Sec:QDP} we discuss quantum differential privacy and how to cast it in terms of the hockey-stick divergence and a variation of the smooth $\max$-relative entropy. Based on this we then give properties and implications of quantum differential privacy including a general bound showing exponential decay of $\delta$ with the algorithm depth. This also gives an operational interpretation of the hockey-stick divergence.
    \item In Section~\ref{Sec:App} we discuss specific noise models including global and local depolarizing noise which we contrast with each other but also with the induced computational limitations. In particular, we get a separation for depolarizing noise where the trace distance decays exponentially, but good differential privacy is reached after a finite number of steps. 
    Furthermore we discuss more general models such as arbitrary qubit noise.
    \item In Section~\ref{Sec:Extensions} We introduce quantum generalisations of local differential privacy and \Renyi differential privacy: two often discussed extensions of the standard differential privacy definition. Finally, we discuss the hypothesis testing interpretation of quantum differential privacy and derive a useful trade-off between $\epsilon$ and $\delta$ from it. 
\end{itemize}

{\noindent \textbf{Notations}}:
A (classical or quantum) system $R$ is associated with a finite-dimensional Hilbert space $\cH_R$. Let $\cPo{R}$ be the set of positive semidefinite linear operators acting on $\cH_R$.
A quantum state $\rho_R$ on $R$ is a positive semidefinite linear operator with unit trace acting on $\cH_R$, denoted $\rho_R\in\cSe{R}$. The set of subnormalized states is denoted $\cSs{R}$. 
A state $\rho_R$ of rank $1$ is called pure, and we may choose a normalized vector $|\psi\rangle_R\in\cH_R$ satisfying  $\rho_R=|\psi\rangle\langle\psi|_R$.
Otherwise, $\rho_R$ is called a mixed state. 
By the spectral theorem, every mixed state can be written as a convex combination of pure states.
For a pure state $|\phi\rangle$ we may use the shorthand $\phi \equiv |\phi\rangle\langle\phi|$.
For a classical system $X$ there is a distinguished orthonormal basis $\lbrace |x\rangle\rbrace_{x=1}^{\dim \cH_X}$ of $\cH_X$ diagonalizing every state on $X$.

A quantum channel $\cN\colon A\to B$ is a linear completely positive and trace-preserving map from the operators on $\cH_A$ to the operators on $\cH_B$.
Given a quantum channel $\mathcal{M}$ with input and output dimension $d$, its Choi matrix is defined as
\begin{align*}
    C_{\mathcal{M}}:=\sum_{i,j\in[d]}\,\mathcal{M}(|i\rangle\langle j|)\otimes |i\rangle\langle j|\,,
\end{align*}
where $\{|i\rangle:\,i\in[d]\}$ is the standard basis of $\mathbb{C}^d$.
In the following we will usually drop the indices as the systems are clear from context. 
A measurement is an operator $0\leq M\leq \Id$ and a collection of measurements such that $\sum_i M_i=\Id$ is called a POVM.

\section{The quantum hockey-stick divergence}\label{Sec:HS-divergence}

In this section we discuss the main technical tools needed for our investigation of quantum differential privacy. 
The quantum hockey-stick divergence was first introduced in~\cite{sharma2012strong}, in the context of exploring strong converse bounds for the quantum capacity, as 
\begin{align}
E_\gamma(\rho\|\sigma) := \tr(\rho-\gamma\sigma)^+, 
\end{align}
for $\gamma\geq 1$. Here $X^+$ denotes the positive part of eigendecomposition of a hermitian matrix $X = X^+ - X^-$. In~\cite{sharma2012strong} it was noted that this quantity is closely related to the trace norm via
\begin{align}
E_\gamma(\rho\|\sigma) = \frac12 \| \rho -\gamma\sigma\|_1 + \frac12(\tr(\rho)-\gamma\tr(\sigma)), \label{Eq:Hs-trace}
\end{align}
so, if $\rho,\sigma\in\cSe{}$, $E_1(\rho\|\sigma) = \frac12 \| \rho -\sigma\|_1$ equals the trace distance. As the trace distance has many desirable properties, we are tempted to hope that similar properties also hold for the quantum hockey-stick divergence. For instance, the trace distance is invariant under unitaries and Eq.~\eqref{Eq:Hs-trace} immediately implies that the same holds true for the hockey-stick divergence.

Another important example will be the following, relating the divergence to a maximization over measurements. 
\begin{lemma}\label{Lem:Dg-meas}
An alternative expression for the hockey-stick divergence for $\rho,\sigma\in\cPo{}$ is given by
\begin{align}
E_\gamma(\rho\|\sigma) = \max_{0\leq\Lambda\leq\Id} \tr\{\Lambda(\rho-\gamma\sigma)\}. 
\end{align}
\end{lemma}
\begin{proof}
The proof is similar to the standard argument for the trace distance. 
Let $X=\rho-\gamma\sigma$ and $(X^+,X^-)$ its  decomposition into positive and negative parts such that $X=X^+ - X^-$. For a general operator $0\leq\Lambda\leq\Id$ one easily sees
\begin{align*}
&\tr\{\Lambda(\rho-\gamma\sigma)\} = \tr\{\Lambda(X^+ - X^-)\} \\ 
&\leq \tr\{\Lambda X^+ \} \leq \tr X^+ = E_\gamma(\rho\|\sigma). 
\end{align*}
It remains to show that equality in the above can be achieved by some measurement. For that simply pick $\Pi_{X^+}$, the projector onto the support of $X^+$ and observe that
\begin{align*}
&\tr\{\Pi_{X^+}(\rho-\gamma\sigma)\} = \tr\{\Pi_{X^+}(X^+ - X^-)\} \\ 
&= \tr\{\Pi_{X^+} X^+ \} = \tr X^+ = E_\gamma(\rho\|\sigma).
\end{align*}
This concludes the proof. 
\end{proof}
A property that was already shown in~\cite{sharma2012strong} is that $E_\gamma$, like any good divergence, obeys data-processing, meaning for any quantum channel $\cN$ we have
\begin{align*}
   E_\gamma(\cN(\rho) \| \cN(\sigma)) \leq  E_\gamma(\rho\|\sigma)\,,
\end{align*}
which even holds for any $\rho,\sigma\in\cPo{}$, see~\cite[Lemma 4]{sharma2012strong}.
This makes it meaningful to define its contraction coefficient as
\begin{align}
\eta_\gamma(\cN) := \sup_{\rho,\sigma\in\mathcal{S}_{=}(\mathcal{H})} \frac{E_\gamma(\cN(\rho) \| \cN(\sigma))}{E_\gamma(\rho\|\sigma)}, 
\end{align} 
where the optimization is over $\rho,\sigma\in\cSe{}$,
and obviously $0\leq \eta_\gamma(\cN) \leq 1$. For a recent overview over contraction coefficients and their properties see~\cite{hirche2020contraction}. 

Interestingly, the contraction coefficient for the trace distance simplifies significantly. Instead of having to optimize over arbitrary initial states $\rho,\sigma$, it suffices to only consider orthogonal pure states~\cite{ruskai1994beyond}. We will prove now that an analogous result holds also for the hockey-stick divergence. This generalizes both the proof for the trace distance in~\cite{ruskai1994beyond} and that for the classical hockey-stick divergence in~\cite[Theorem 3]{asoodeh2020privacy}. 
\begin{theorem}\label{Thm:ConCoef-easy}
The hockey-stick divergence contraction coefficient can be equivalently expressed as
\begin{align}
\eta_\gamma(\cN) = \sup_{|\varphi\rangle\perp |\psi\rangle} E_\gamma(\cN(|\varphi\rangle\langle \varphi|)\|\cN(|\psi\rangle\langle \psi|)). 
\end{align}
\end{theorem}
\begin{proof}
Let $X=\rho-\gamma\sigma$ and $(X^+,X^-)$ its  decomposition into positive and negative parts such that $X=X^+ - X^-$. Note that by definition and Equation~\eqref{Eq:Hs-trace} we have
\begin{align}
E_\gamma(\rho\|\sigma) = \tr X^+ &= \frac12 \| X \|_1 + \frac12(1-\gamma),\nonumber \\
E_\gamma(\cN(\rho) \| \cN(\sigma)) &= \frac12 \| \cN(X) \|_1 + \frac12(1-\gamma).  \label{Eq:Dgamma-N}
\end{align}
Further, define $\hat{X}^+ = \frac{X^+}{\tr X^+}$ and $\hat{X}^- = \frac{X^-}{\tr X^-}$ with spectral decompositions 
\begin{align*}
\hat{X}^+ &= \sum_m p_m |m\rangle \langle m|, \\
\hat{X}^- &= \sum_n q_n |n\rangle \langle n|.
\end{align*}
With these definitions, we observe
\begin{align}
&\| \cN(X)\|_1 \\
&= \| \cN(X^+) - \cN(X^-) \|_1  \nonumber\\
&= \| \tr(X^+)\cN(\hat{X}^+) - \tr(X^-)\cN(\hat{X}^-) \|_1  \nonumber\\
&= \left\| \tr(X^+)\cN(\sum_m p_m |m\rangle \langle m|)\right. \nonumber\\
&\qquad \left. -  \tr(X^-)\cN(\sum_n q_n |n\rangle \langle n|) \right\|_1  \nonumber\\
&= \left\| \sum_{m,n} p_m q_n \left( \tr(X^+)\cN( |m\rangle \langle m|) \right.\right. \nonumber\\
&\qquad- \left. \vphantom{\sum_{m,n}}\left. \tr(X^-)\cN( |n\rangle \langle n|)\right) \right\|_1 \nonumber \\
&\leq \sum_{m,n} p_m q_n \left\| \tr(X^+)\cN( |m\rangle \langle m|) \right. \nonumber\\
&\qquad- \left.\tr(X^-)\cN( |n\rangle \langle n|) \right\|_1  \nonumber\\
&\leq \sup_{m,n} \| \tr(X^+)\cN( |m\rangle \langle m|) - \tr(X^-)\cN( |n\rangle \langle n|) \|_1. \label{Eq:max-nm}
\end{align}
We continue with
\begingroup
\allowdisplaybreaks
\begin{align*}
&\| \tr(X^+)\cN( |m\rangle \langle m|) - \tr(X^-)\cN( |n\rangle \langle n|) \|_1 \\
\leq& \| \tr(X^+)\cN( |m\rangle \langle m|) -  \gamma\tr(X^+)\cN( |n\rangle \langle n|) \|_1 \\
&+ \| \gamma\tr(X^+)\cN( |n\rangle \langle n|) - \tr(X^-)\cN( |n\rangle \langle n|) \|_1   \\
=& \tr(X^+) \| \cN( |m\rangle \langle m|) -  \gamma\cN( |n\rangle \langle n|) \|_1 \\
&+ \| ((\gamma-1)\tr(X^+) + (1-\gamma))\cN( |n\rangle \langle n|)  \|_1  \\
=& \tr(X^+) \| \cN( |m\rangle \langle m|) -  \gamma\cN( |n\rangle \langle n|) \|_1 \\
&- ((\gamma-1)\tr(X^+) + (1-\gamma))  \\
=& \tr(X^+) ( \| \cN( |m\rangle \langle m|) -  \gamma\cN( |n\rangle \langle n|) \|_1 \\ 
&- (\gamma-1)) - (1-\gamma) \\
=& 2 \tr(X^+) E_\gamma(\cN( |m\rangle \langle m|) \| \cN( |n\rangle \langle n|)) - (1-\gamma) \\
=& 2 E_\gamma(\rho\|\sigma) E_\gamma(\cN( |m\rangle \langle m|) \| \cN( |n\rangle \langle n|)) - (1-\gamma),  
\end{align*}
\endgroup
where the first inequality is by triangle inequality, the first equality because $\tr(X^+)\geq 0$ and the second because $((\gamma-1)\tr(X^+) + (1-\gamma))\leq 0$ since $\tr(X^+)=E_\gamma(\rho\|\sigma)\le 1$. The remaining steps consist of shuffling terms and applying definitions. 
Plugging this back into Equation~\eqref{Eq:max-nm}, we get
\begin{align*}
&\| \cN(X)\|_1  \\
&\leq \sup_{m,n} 2 E_\gamma(\rho\|\sigma) E_\gamma(\cN( |m\rangle \langle m|) \| \cN( |n\rangle \langle n|)) - (1-\gamma)
\end{align*}
and again plugging this into Equation~\eqref{Eq:Dgamma-N}, we have
\begin{align*}
&E_\gamma(\cN(\rho) \| \cN(\sigma)) \\
&\leq \frac12 \left\{\sup_{m,n} 2 E_\gamma(\rho\|\sigma) E_\gamma(\cN( |m\rangle \langle m|) \| \cN( |n\rangle \langle n|)) \right. \\
&\left.\vphantom{\sup_{m,n} 2 E_\gamma} \quad - (1-\gamma)\right\} + \frac12(1-\gamma) \\
&=   E_\gamma(\rho\|\sigma) \sup_{m,n} E_\gamma(\cN( |m\rangle \langle m|) \| \cN( |n\rangle \langle n|)) .
\end{align*}
From here it follows directly that
\begin{align*}
\eta_\gamma(\cN) \leq \sup_{|\varphi\rangle \perp |\psi\rangle} E_\gamma(\cN(|\varphi\rangle\langle \varphi|)\|\cN(|\psi\rangle\langle \psi|)). 
\end{align*}
It remains to show that the inequality is indeed achieved. For that simply note that for all orthogonal $|\varphi\rangle$, $|\psi\rangle$, we have 
\begin{align*}
E_\gamma (|\varphi\rangle\langle \varphi| \| |\psi\rangle\langle \psi|) = 1
\end{align*}
and therefore 
\begin{align*}
\eta_\gamma(\cN) &\geq  \sup_{|\varphi\rangle \perp |\psi\rangle} \frac{E_\gamma(\cN(|\varphi\rangle\langle \varphi|)\|\cN(|\psi\rangle\langle \psi|))}{E_\gamma (|\varphi\rangle\langle \varphi| \| |\psi\rangle\langle \psi|) }  \\
&= \sup_{|\varphi\rangle \perp |\psi\rangle} E_\gamma(\cN(|\varphi\rangle\langle \varphi|)\|\cN(|\psi\rangle\langle \psi|)), 
\end{align*}
where the lower bound follows from the fact that we are restricting the supremum to a smaller set than in the definition of the contraction coefficient. Furthermore, for orthogonal pure states one can check that $E_\gamma (|\varphi\rangle\langle \varphi| \| |\psi\rangle\langle \psi|)=1$. This concludes the proof.
\end{proof}

Next we prove a Fuchs-van-de-Graaf type inequality that reduces to the well-known 
\begin{align*}
    \frac12 \| \rho-\sigma\|_1 \leq \sqrt{1-F(\rho,\sigma)}
\end{align*}
for $\gamma=1$, where $F(\rho,\sigma):=\|\sqrt{\rho}\sqrt{\sigma}\|_1^2$ is the quantum fidelity. 
\begin{lemma}\label{Lem:HS-FvdG}
For $\gamma\geq 1$ and $\rho,\sigma\in\cPo{}$, we have
\begin{align}
    &E_\gamma(\rho\|\sigma) \\ 
    &\leq \frac12\sqrt{(\tr(\rho+\gamma\sigma))^2 - 4\gamma (\tr\rho)^2(\tr\sigma)^2 F(\hat\rho,\hat\sigma)} \nonumber\\
    &\quad+ \frac{\tr(\rho-\gamma\sigma)}{2}\,,
\end{align}
where $\hat\rho=\frac{\rho}{\tr\rho}$ and $\hat\sigma=\frac{\sigma}{\tr\sigma}$. 
For $\rho,\sigma\in\cSe{}$ this simplifies to
\begin{align}
    E_\gamma(\rho\|\sigma) \leq \frac12\sqrt{(1+\gamma)^2 - 4\gamma F(\rho,\sigma)} + \frac{(1-\gamma)}{2}\,.
\end{align}
\end{lemma}
\begin{proof}
Using~\cite[Supplementary~Lemma~3]{CKW14} as stated in Lemma~\ref{lem:FvdG-PSD} we get
\begin{align*}
    \|\rho-\gamma\sigma\|_1^2 + 4\|\sqrt{\rho}\sqrt{\gamma\sigma}\|_1^2 \leq (\tr[\rho+\gamma\sigma])^2,
\end{align*}
which immediately implies 
\begin{align*}
    &\|\rho-\gamma\sigma\|_1 \\
    &\leq \sqrt{(\tr(\rho+\gamma\sigma))^2 - 4\gamma (\tr\rho)^2(\tr\sigma)^2 F(\hat\rho,\hat\sigma)}\,. 
\end{align*}
Plugging this into Equation~\eqref{Eq:Hs-trace} gives the desired result. 
\end{proof}

This also gives us a bound on the contraction coefficient as
\begin{align}\label{Eq:eta-FvdG}
&\eta_\gamma(\cN) \\
&\leq \sup_{|\varphi\rangle\perp |\psi\rangle} \frac12\sqrt{(1+\gamma)^2 - 4\gamma F(\cN(|\varphi\rangle\langle \varphi|),\cN(|\psi\rangle\langle \psi|))} \nonumber\\
&\quad+ \frac{(1-\gamma)}{2}. 
\end{align}

We can also bound the hockey-stick divergence for any $\gamma\geq 1$ directly by the trace-distance, leading to an alternative way of bounding the contraction coefficients. This generalizes the analog classical result in~\cite{asoodeh2021local}. 
\begin{lemma}\label{Lem:gamma-trace-distance}
For $\gamma\geq 1$ and $\rho,\sigma\in\cSe{}$, we have
\begin{align}
    1-\gamma(1-\frac12\|\rho-\sigma\|_1) \leq E_\gamma(\rho\|\sigma) \leq \frac12\|\rho-\sigma\|_1,
\end{align}
implying
\begin{align}
    1-\gamma(1-\eta_1(\cN)) \leq \eta_\gamma(\cN) \leq \eta_1(\cN),
\end{align}   
\end{lemma}
\begin{proof}
The only thing needed to prove is the first inequality and the remaining statement follows easily. To that end we observe
\begin{align*}
    &\gamma \frac12\|\rho-\sigma\|_1 \\
    &= \gamma \max_{0\leq\Lambda\leq\Id} \tr\Lambda(\rho-\sigma) \\
    &=\max_{0\leq\Lambda\leq\Id} \tr\Lambda(\gamma\rho-\gamma\sigma + \rho - \rho) \\
    &\leq\max_{0\leq\Lambda\leq\Id} \tr\Lambda(\rho-\gamma\sigma) + \max_{0\leq\Lambda\leq\Id}\tr\Lambda(\gamma\rho - \rho) \\
    &= E_\gamma(\rho\|\sigma) +\gamma -1. 
\end{align*}
The second inequality of the statement follows by definition and the third and fourth by using the simplified expression for the contraction coefficients in Theorem~\ref{Thm:ConCoef-easy}. 
\end{proof}

Finally, we list some potentially useful properties of $E_\gamma$. Some of these are generalizations of classical results proven in~\cite{liu2016e}. 
\begin{proposition}\label{Lem:HSD-properties}
We have the following properties:
\begin{itemize}
 \item{(Triangle inequality)}   For $\gamma_1, \gamma_2 \geq 1$ and $\rho,\sigma\in\cPo{}$, we have
\begin{align}
E_{\gamma_1\gamma_2}(\rho\|\sigma) \leq E_{\gamma_1}(\rho\|\tau) + \gamma_1 E_{\gamma_2}(\tau\|\sigma). \label{Eq:Triangle}
\end{align}
\item{(Strong convexity)} Let $\gamma_1,\gamma_2\geq 1$, $\rho=\sum_x p(x)\rho_x$ and  $\sigma=\sum_x q(x)\sigma_x$ with $\rho_x,\sigma_x\in\cPo{}$, we have
\begin{align}
    E_{\gamma_1\gamma_2}(\rho\|\sigma) \leq \sum_x p(x) E_{\gamma_1}(\rho_x\|\sigma_x) + \gamma_1  E_{\gamma_2}(\tilde p\| \tilde q), 
\end{align}
where $\tilde p$ and $\tilde q$ are non-normalized  distributions $\tilde p(x) = p(x)\tr\sigma_x$ and $\tilde q(x) = q(x)\tr\sigma_x$, respectively. This also implies convexity and joint convexity. 
\item{(Stability)} For $\gamma\geq 1$ and $\rho,\sigma,\tau\in\cPo{}$, $\tau\not=0$ , we have
\begin{align}
    E_\gamma(\rho\otimes\tau\|\sigma\otimes\tau) = \tr\left[\tau\right]E_\gamma(\rho\|\sigma).
\end{align}
\item{(Subadditivity)} For $\gamma_1,\gamma_2\geq 1$ and $\rho_1,\sigma_1\in\cP(\cH_1)$, $\rho_2,\sigma_2\in\cP(\cH_2)$, we have
\begin{align}\label{equ:subadditivity_gamma}
    &E_{\gamma_1\gamma_2}(\rho_1\otimes\rho_2\|\sigma_1\otimes\sigma_2) \nonumber\\
    &\leq\tr[\rho_2]E_{\gamma_1}(\rho_1\|\sigma_1) + \tr[\sigma_1]\gamma_1 E_{\gamma_2}(\rho_2\|\sigma_2), \\
    &E_{\gamma_1\gamma_2}(\rho_1\otimes\rho_2\|\sigma_1\otimes\sigma_2) \nonumber\\
    &\leq \tr[\rho_1]E_{\gamma_1}(\rho_2\|\sigma_2) + \tr[\sigma_2]\gamma_1 E_{\gamma_2}(\rho_1\|\sigma_1) .
\end{align}
\item{(Symmetry)} For $\gamma\geq 1$ and $\rho,\sigma\in\cPo{}$, we have
\begin{align}
    E_\gamma(\rho\|\sigma) = \gamma E_{\frac1\gamma}(\sigma\|\rho) + (\tr(\rho)-\gamma\tr(\sigma)).
\end{align}
\item{(Trace bound)} For $\gamma\geq 1$ and $\rho,\sigma,\tau\in\cPo{}$, we have
\begin{align}
    &E_\gamma(\rho\|\sigma) + E_\gamma(\rho\|\tau) \\
    &\geq \frac{\gamma}{2}\|\tau - \sigma\|_1 +(\tr(\rho)-\gamma\tr(\tau)).
\end{align}
\end{itemize}
\end{proposition}
\begin{proof}
The first statement is easily seen as follows, 
\begin{align*}
  &E_{\gamma_1\gamma_2}(\rho\|\sigma) \\
  &= \max_{0\leq\Lambda\leq\Id} \tr\Lambda(\rho-\gamma_1\gamma_2\sigma) \\
  &= \max_{0\leq\Lambda\leq\Id} \tr\Lambda(\rho-\gamma_1\tau + \gamma_1\tau  -  \gamma_1\gamma_2\sigma) \\
  &\leq \max_{0\leq\Lambda\leq\Id} \tr\Lambda(\rho-\gamma_1\tau) + \max_{0\leq\Lambda\leq\Id} \tr\Lambda(\gamma_1\tau  -  \gamma_1\gamma_2\sigma) \\
  &\leq E_{\gamma_1}(\rho\|\tau) + \gamma_1 E_{\gamma_2}(\tau\|\sigma).
\end{align*}
The strong convexity follows similarly by considering $\tau=\sum_x p(x) \sigma_x$. 
Stability follows by first considering the statement for the state $\tau'=\tau/\tr\left[\tau\right]$. Applying data-processing twice, once for the partial trace and once for $\cN(\rho)=\rho\otimes\tau'$ gives the statement for states. The generalization then follows by noting that the underlying quantity is absolutely homogeneous in $\tau$.
 Subadditivity follows by triangle inequality and stability, picking first $\tau=\sigma_1\otimes\rho_2$ and then $\tau=\rho_1\otimes\sigma_2$. For the last two items we need the observation that 
\begin{align*}
    \max_{0\leq\Lambda\leq\Id} \tr\Lambda(\rho-\gamma\sigma) = \max_{0\leq\Lambda\leq\Id} \tr(\Id - \Lambda) (\rho-\gamma\sigma). 
\end{align*}
From this symmetry follows immediately and the trace bound by observing
\begin{align*}
    &E_\gamma(\rho\|\sigma) + E_\gamma(\rho\|\tau) \\
    &= \max_{0\leq\Lambda\leq\Id} \tr\Lambda(\rho-\gamma\sigma) + \max_{0\leq\Lambda\leq\Id} \tr(\Id - \Lambda) (\rho-\gamma\tau) \\
    &\geq \max_{0\leq\Lambda\leq\Id} \tr\Lambda(\rho-\gamma\sigma - \rho + \gamma\tau) +(\tr(\rho)-\gamma\tr(\tau)) \\
    &= \frac{\gamma}{2}\|\tau - \sigma\|_1 +(\tr(\rho)-\gamma\tr(\tau)). 
\end{align*}
\end{proof}
Next we connect the hockey-stick divergence to the smooth $\max$-relative entropy. This is similar to~\cite[Lemma 1]{zhou2017differential}, however based on different definitions. Also, besides avoiding explicit use of measurements, we provide a considerably simpler proof. For a classical analogue of both results, see~\cite{Dwork2013,liu2016e}. 
\begin{lemma}\label{Lem:max-rel-ent}
For $\gamma\geq1$ we have, 
\begin{align}\label{equ:dmax-hockey}
    D_{\max}^\epsilon(\rho\|\sigma)\leq\log\gamma \quad\Leftrightarrow\quad E_\gamma(\rho\|\sigma)\leq\epsilon, 
\end{align}
where $D_{\max}^\epsilon(\rho\|\sigma) = \inf_{\bar\rho\in B^\epsilon(\rho)} D_{\max}(\bar\rho\|\sigma)$ is the smooth $\max$-relative entropy with $D_{\max}(\rho\|\sigma)=\inf\{\lambda : \rho\leq e^{\lambda}\sigma\}$ and $B^\epsilon(\rho) = \{\bar\rho : \bar\rho\in\cPo{} \wedge E_1(\rho,\bar\rho)\leq \epsilon\}$.
\end{lemma}
\begin{proof}
First, we show the $\Rightarrow$ direction. Assume that $D_{\max}^\epsilon(\rho\|\sigma)\leq\log\gamma$, this implies that there exists a $\bar\rho$ such that $E_1(\rho,\bar\rho)\leq\epsilon$ and $\bar\rho\leq\gamma\sigma$. From this and the triangle inequality in \Cref{Lem:HSD-properties} we immediately get
\begin{align*}
    E_\gamma(\rho\|\sigma) &\leq E_1(\rho\|\bar\rho) + E_\gamma(\bar\rho\|\sigma) \\ 
    &\leq \epsilon + 0. 
\end{align*}
We now show the $\Leftarrow$ direction. Fix $\bar\rho=\gamma\sigma$. As $\gamma\geq1$ and $\sigma$ is a positive operator, this is a positive operator.
Furthermore, we have that
\begin{align*}
    E_1(\rho,\bar\rho) &= \tr(\rho-\bar\rho)^+ \\
    &= \tr(\rho-\gamma\sigma)^+ \\
    &= E_\gamma(\rho,\sigma)\leq \epsilon\,, 
\end{align*}
where in the last inequality we used our hypothesis. Thus, $\bar{\rho}\in B^{\epsilon}(\rho)$. Now observe that $D_{\max}(\bar\rho\|\sigma)=\log\gamma$ by construction, from which it follows that $D_{\max}^\epsilon(\rho\|\sigma)\leq\log\gamma$.
We therefore see that $E_\gamma(\rho,\sigma)\leq\epsilon$ implies $D_{\max}^\epsilon(\rho\|\sigma)\leq\log\gamma$. 
\end{proof}

Note that $D_{\max}^\epsilon$ does not correspond to the usual definition of the smooth max-relative entropy, as we do not constrain the optimization to normalized or subnormalized operators and use $E_1$ as our distance measure of choice. It does however generalize the definition used in the proof of the analog classical lemma, see~\cite[Definition 8]{liu2016e}. Also,  
\cite[Lemma 1]{zhou2017differential} proves a similar statement optimizing over normalized operators that are however not necessarily positive. 
It remains open for now whether both conditions can be achieved simultaneously in the above proof.

In the next section, we will proceed to apply the above results to quantum differential privacy.

\section{Quantum differential privacy}\label{Sec:QDP}

There are several similar definitions of quantum differential privacy in the literature that apply to different settings. 
Following~\cite{zhou2017differential}, we will formulate ours in a way that it applies to an arbitrary quantum algorithm $\cA$, i.e. a completely positive, trace-preserving map. The general idea is that if we apply the algorithm to a state from a fixed database, say $\cD$, then a malicious party gaining access to the output should not be able to distinguish by any measurement whether the used input was a certain state or one of its immediate neighbors in the database. Classically the motivation is usually to consider a set of databases that are neighbors if they differ in only one entry, e.g. an observer of a medical trial should not be able to determine the results of an individual participant. 

In the quantum literature, several definitions of the neighboring status are used, e.g. closeness in trace distance or reachability by a single local operation. Leaving the exact choice of definition open for now, we denote two states being neighbors by $\sigma\sim\rho$. We can now state the definition of quantum differential privacy. 

\begin{definition}\label{Def:QDP}
Let $\cD$ be a set of quantum states and $\cA$ be a quantum algorithm (i.e. a $\operatorname{CPTP}$ map). We call $\cA$ $(\epsilon,\delta)$-differentially private if for all measurements $0\leq M \leq \Id$ and all $\rho, \sigma\in\cD$ such that $\rho\sim\sigma$, we have
\begin{align}
\tr( M \cA(\rho)) \leq e^\epsilon \tr( M \cA(\sigma)) + \delta. \label{Eq:DP-def}
\end{align}
We simply call $\cA$ $\epsilon$-differentially private if $\cA$ is $(\epsilon,0)$-differentially private. 
\end{definition}
This definition is a rather direct generalization of classical differential privacy. Indeed, if we only consider diagonal states and projectors $P$ in the computational basis in the definition above and interpret $\tr( P \cA(\sigma))$ as the probability of the measurement outcome lying in a given set, we obtain exactly the definition in~\cite[Definition 2.4]{dwork2006differential}.
But for quantum states we need to optimize over all possible basis and can even consider POVMs. However, we will now use the tools from the previous section to see that it can indeed be checked as a property of the quantum states themselves without explicitly considering the measurements. 

\begin{lemma}\label{Lem:DP-divergence} 
The following three statements are equivalent, 
\begin{align}
&\cA\, \text{is $(\epsilon,\delta)$-differentially private} \\&\quad\Leftrightarrow\quad \sup_{\rho\sim\sigma} E_{e^\epsilon}(\cA(\rho) \| \cA(\sigma) ) \leq \delta \\
&\quad\Leftrightarrow\quad \sup_{\rho\sim\sigma} D_{\max}^{\delta}(\cA(\rho) \| \cA(\sigma) ) \leq \epsilon. \label{Eq:DP-max-rel-ent}
\end{align}
\end{lemma} 
\begin{proof}
Note that we can rewrite the condition in Equation~\eqref{Eq:DP-def} as
\begin{align*}
\tr\left( M (\cA(\rho) -  e^\epsilon \cA(\sigma))\right) \leq \delta. 
\end{align*}
Since this has to hold for all measurements and all neighboring input states we can use Lemma~\ref{Lem:Dg-meas} to conclude the first equivalence. The second then follows directly from Lemma~\ref{Lem:max-rel-ent}.
\end{proof}
Note that Lemma~\ref{Lem:DP-divergence} implies that if all the outputs of an algorithm are diagonal in the same basis, then quantum DP is equivalent to classical DP. This is because for two states that commute, $E_\gamma$ is the same for the states and the corresponding probability distributions.
Note that in the case of $\epsilon=0$ we have a condition on the trace distance. This is also known as local sensitivity as e.g. defined in~\cite[Definition 7.1]{Dwork2013}, which in turn is closely related to the stability of an algorithm, see~\cite[Section 7.3]{Dwork2013}. For recent quantum generalization of the latter concept see also~\cite{arunachalam2021private}.  

The above allows us to immediately conclude some well known properties, however with remarkably simple proofs solely based on properties of the divergences.
\begin{corollary}\label{cor:post-parallel}
The following properties hold.
\begin{itemize}
    \item{(Post-processing)} Let $\cA$ be $(\epsilon,\delta)$-differentially private and $\cN$ be an arbitrary quantum channel, then $\cN\circ\cA$ is also $(\epsilon,\delta)$-differentially private. 
    \item{(Parallel composition)} Let $\cA_1$ be $(\epsilon_1,\delta_1)$-differentially private and $\cA_2$ be $(\epsilon_2,\delta_2)$-differentially private. Define that $\rho_1\otimes\rho_2\sim\sigma_1\otimes\sigma_2$ if $\rho_1\sim\sigma_1$ and $\rho_2\sim\sigma_2$. Then $\cA_1\otimes\cA_2$ is $(\epsilon_1+\epsilon_2, \bar\delta)$-differentially private on such product states, with $\bar\delta=\min\{\delta_1+e^{\epsilon_1}\delta_2, e^{\epsilon_2}\delta_1+\delta_2\}$. 

\end{itemize}
\end{corollary}
\begin{proof}
The post-processing property was first shown in~\cite[Proposition 1]{zhou2017differential} and composition in~\cite[Theorem 4]{zhou2017differential} for $\delta_1=\delta_2=0$ and for the general case in~\cite[Theorem 4]{zhou2017differential}, although the latter relied on a erroneous assumption in~\cite[Lemma 1]{zhou2017differential}, see below. Now, post-processing simply follows by data-processing of the hockey-stick divergence. Composition for $\delta_1=\delta_2=0$ is of course implied by the general case, however also follows easily by subadditivity of the hockey-stick divergence which we showed in Lemma~\ref{Lem:HSD-properties}. The general case can be seen from the differential privacy formulation in Equation~\eqref{Eq:DP-max-rel-ent}. Let $\bar\rho_1$ and $\bar\rho_2$ be the optimizers in the smooth $\max$-relative entropies implied by the differential privacy assumption. One easily sees that 
\begin{align}\label{equ:tensorproduct_positive}
\bar\rho_1\otimes\bar\rho_2 \leq e^{\epsilon_1+\epsilon_2}\cA_1(\sigma_1)\otimes\cA_2(\sigma_2)
\end{align}
and 
\begin{align*}    &E_1(\cA_1(\rho_1)\otimes\cA_2(\rho_2)\|\bar\rho_1\otimes\bar\rho_2) \\
    &\leq  \tr[\bar\rho_2] E_1(\cA_1(\rho_1)\|\bar\rho_1) +  \tr[\cA_1(\rho_1)]E_1(\cA_2(\rho_2)\|\bar\rho_2) \\
    &\leq e^{\eps_2}\delta_1 + \delta_2, 
\end{align*}
which follows from the subadditivity property in Proposition~\ref{Lem:HSD-properties} and similarly, 
\begin{align*}  E_1(\cA_1(\rho_1)\otimes\cA_2(\rho_2)\|\bar\rho_1\otimes\bar\rho_2)  \leq \delta_1 + e^{\eps_1}\delta_2. 
\end{align*} 
This makes $\bar\rho=\bar\rho_1\otimes\bar\rho_2$ a valid choice to prove the claim. 
\end{proof}
We note that the parallel composition property was also claimed in~\cite[Theorem 5]{zhou2017differential}. However, note that in their analogue of Eq.~\eqref{equ:dmax-hockey} in~\cite[Lemma 1]{zhou2017differential} their definition of $D_{\max}^\epsilon$ requires an optimization over states which are not necessarily positive. As Eq.~\eqref{equ:tensorproduct_positive} does not hold for operators that are not positive (i.e. $A_1\leq B_1, A_2\leq B_2\nRightarrow A_1\otimes A_2\leq B_1\otimes B_2$ in general), the parallel decomposition does not follow. Thus, to the best of our knowledge, Corollary~\ref{cor:post-parallel} is the first to establish parallel composition property for quantum differential privacy.

The implementation of many near-term quantum algorithms on noisy devices can be modelled by layers of intended channels $\cC_i$, e.g. gates in a circuit or layers of a quantum neural network, directly followed by intermediate noise $\cN_i$, i.e. 
\begin{align}
\cA = \bigcirc_i^n \cN_i \circ \cC_i. \label{Eq:LayerA}
\end{align}
These types of algorithms are predestined to applying an approach based on contraction coefficients. 
Doing so, we get the following result. 
\begin{proposition}\label{Thm:Contraction}
For an algorithm $\cA$ of the form in Equation~\eqref{Eq:LayerA} we have
\begin{align}
E_{e^\epsilon}(\cA(\rho) \| \cA(\sigma) ) \leq \left(\prod_i \eta_{e^\epsilon}(\cN_i) \right) E_{e^\epsilon}(\rho \| \sigma )\,.
\end{align}
\end{proposition}
\begin{proof} 
The proof follows by alternatingly using the definition of the contraction coefficient to shell off the noise and using data processing to remove the computational layers. 
\end{proof}
This directly implies a decay in the $\delta$ parameter for differential privacy based on the contraction coefficient of the noise channels. This intuition becomes even more transparent when we consider a special case. 
\begin{corollary}\label{Cor:delta-decay}
For an algorithm $\cA$ of the form in Equation~\eqref{Eq:LayerA} with all $\cN_i=\cN$ identical and $\rho\sim\sigma$ if $\frac12 \|\rho-\sigma\|_1 \leq\kappa$, then $\cA$ is $(\epsilon, \delta)$-differentially private with
\begin{align}
\delta = \left(\eta_{e^\epsilon}(\cN)\right)^n \kappa. \label{Eq:delta-decay}
\end{align}
\end{corollary}
\begin{proof}
The corollary is a direct application of Proposition~\ref{Thm:Contraction}. 
\end{proof}
This implies in particular an exponential decay of $\delta$ with the length of the algorithm, i.e. every such algorithm with $\eta_{e^\epsilon}(\cN)<1$ will eventually become differentially private with vanishing $\delta$ for large enough $n$. Note that generally $\eta_{e^\epsilon}(\cN)$ is a function of the channel but also $\epsilon$ allowing for a certain trade-off between $\epsilon$ and $\delta$. Finally, we remark that thanks to Theorem~\ref{Thm:ConCoef-easy} in the previous section, $\eta_{e^\epsilon}(\cN)$ is more easily computable and can be controlled analytically or numerically for many noise models. Before moving to more specific models, we give a simple bound on the measured relative entropy, defined as
\begin{align}
     D_M(\rho \| \sigma) = \sup_{\{M_x\}}  D( \tr(M_x \rho) \| \tr(M_x \sigma)), 
\end{align}
of a differentially private channel. The relative entropy is a common tool when comparing quantum states and therefore the result might be of independent interest.  
\begin{lemma}\label{Lem:meas-rel-ent}
Let $\cA$ be $\epsilon$-differentially private. Then for all $\rho\sim\sigma$,
\begin{align}
    D_M(\cA(\rho) \| \cA(\sigma)) \leq 2\epsilon E_1(\cA(\rho) \| \cA(\sigma)) \leq 2\epsilon(1-e^{-\epsilon}). 
\end{align}
\end{lemma}
\begin{proof}
Let $\{M_x\}$ be the POVM that achieves the maximum in $D_M(\cA(\rho) \| \cA(\sigma))$ and let $p_x=\tr M_x\cA(\rho)$ and $q_x=\tr M_x\cA(\sigma)$. Then, 
\begin{align*}
    D_M(\cA(\rho) \| \cA(\sigma)) &= \sum_x p_x \log\frac{p_x}{q_x} \\
    &\leq \sum_x p_x \log\frac{p_x}{q_x} +  \sum_x q_x \log\frac{q_x}{p_x} \\
    &= \sum_x (p_x-q_x) \log\frac{p_x}{q_x} \\
    &\leq \sum_x (p_x-q_x) \epsilon \\
    &= \epsilon\sum_x \tr M_x(\cA(\rho)-\cA(\sigma)) \\
    &\leq 2\epsilon E_1(\cA(\rho) \| \cA(\sigma)) \\
    &\leq 2\epsilon\,(1-e^{-\epsilon})\,,
\end{align*}
where the first equality is by definition and the second and third obvious. The first inequality follows because the relative entropy is non-negative, the second by $\epsilon$-DP, the third by a property of the trace distance and the final one by Lemma~\ref{Lem:gamma-trace-distance} and again $\epsilon$-DP. 
\end{proof}
This lemma gives a quantum version of several similar relations for classical differential privacy. We remark that it might be possible to find tighter bounds along the lines of the classical~\cite[Theorem 1]{duchi2013local}. We leave investigating non-measured quantities for future work, however remark here that because 
\begin{align*}
    D(\cA(\rho) \| \cA(\sigma)) \leq D_{\max}(\cA(\rho) \| \cA(\sigma)), 
\end{align*}
$\epsilon$-differential privacy always also implies small relative entropy for neighbouring input states and refer to Section~\ref{Sec:Renyi} for relaxations involving \Renyi relative entropies.

\section{Applications to specific noise models}\label{Sec:App}

In the remainder of this section we will discuss particular examples and implications of the above results, including global and local depolarizing noise as well as arbitrary local qubit noise.

\subsection{Global depolarizing noise}

\begin{figure*}
    \centering
    \begin{minipage}{0.47\textwidth}
        \centering
        \includegraphics[width=1\textwidth]{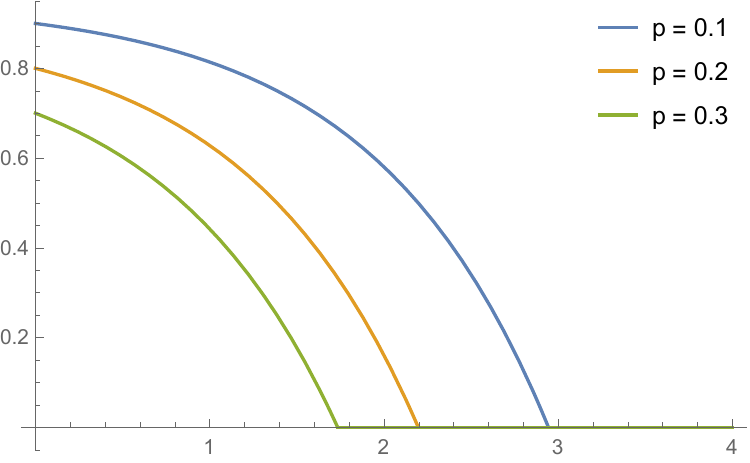} %
    \end{minipage}\hfill
    \begin{minipage}{0.47\textwidth}
        \centering
        \includegraphics[width=1\textwidth]{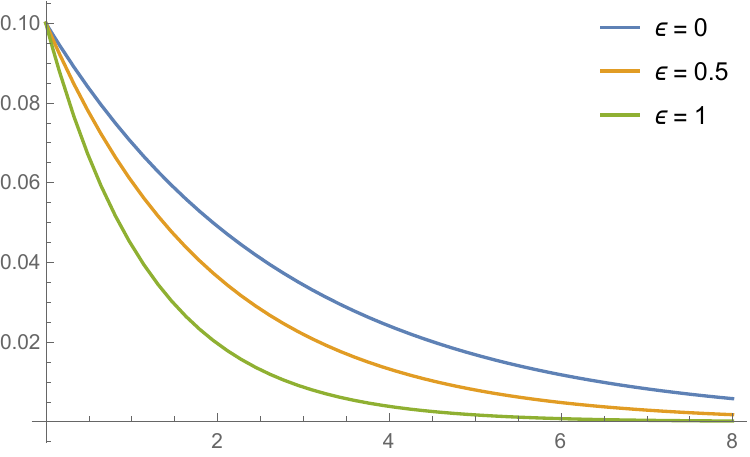} %
    \end{minipage}
\caption{\label{Fig:Depolarizing} Left: the contraction coefficient $\eta_{e^\epsilon}(\cD_p)$ for $D=2$ and different values of $p$ plotted against $\epsilon$. Right: the exponential decay given in Equation~\eqref{Eq:delta-decay} for a depolarizing channel with $p=0.3$, $\kappa=0.1$ and different values of $\epsilon$ plotted against $n$, the number of layers with the given depolarizing noise.}
\end{figure*}

A typical noise channel is the depolarizing channel defined as
\begin{align}
    \cD_p(\rho)=(1-p)\rho +p\frac{\Id}{D},
\end{align}
with $0\leq p\leq 1$ and $D$ the dimension of the system. In this case we can easily bound the contraction coefficient.
\begin{lemma}\label{Lem:HS-Dp-bounds}
For $0\leq p\leq 1$ and $\gamma\geq1$ we have
\begin{align}
        &E_\gamma(\cD_p(\rho)\|\cD_p(\sigma)) \\
        &\leq \max\{0,(1-\gamma)\frac{p}{D} + (1-p)E_\gamma(\rho\|\sigma)\} \label{Eq:Dp-Eg-bound}
\end{align}
and 
\begin{align}
    \eta_\gamma(\cD_p) = \max\{0,(1-\gamma)\frac{p}{D} + (1-p)\}. \label{Eq:Dp-cc}
\end{align}
\end{lemma}
\begin{proof}
Note that
\begin{align*}
    &E_\gamma(\cD_p(\rho)\|\cD_p(\sigma)) \\
    &= \tr((1-\gamma)p\frac{\Id}{D} + (1-p)(\rho-\gamma\sigma))^+ \\
    &= \tr P^+((1-\gamma)p\frac{\Id}{D} + (1-p)(\rho-\gamma\sigma)), 
\end{align*}
where $P^+$ is the projector onto the positive subspace of $(1-\gamma)p\frac{\Id}{D} + (1-p)(\rho-\gamma\sigma)$. Observe that 
\begin{align*}
    E_\gamma(\cD_p(\rho)\|\cD_p(\sigma))>0 \quad\Rightarrow\quad \tr P^+\geq 1.
\end{align*}
Considering this case we get
\begin{align*}
    &E_\gamma(\cD_p(\rho)\|\cD_p(\sigma)) \\
    &= (1-\gamma)\frac{p}{D}\tr P^+ + (1-p)(\tr P^+(\rho-\gamma\sigma)) \\
    &\leq (1-\gamma)\frac{p}{D} + (1-p)E_\gamma(\rho\|\sigma) \\
    &\leq (1-\gamma)\frac{p}{D} + (1-p). 
\end{align*}
Note that for sufficiently large $\gamma$ the upper bound could become negative, but one can easily check that in this case $E_\gamma(\cD_p(\rho)\|\cD_p(\sigma))=0$ implying that we are in the other case. 
It remains to show that this is optimal but this is easy to see by picking any two orthogonal pure states.
\end{proof}

In Figure~\ref{Fig:Depolarizing} we show an example of the contraction coefficient for the depolarizing channel and examples for the exponential decay of $\delta$ in Equation~\eqref{Eq:delta-decay} for different values of $\epsilon$ using the contraction coefficient in Equation~\ref{Eq:Dp-cc}. It becomes clear that while iterative algorithms with depolarizing layers lead to strong privacy, using the contraction coefficient bound, there is only little space before the output states also become mostly useless (the case $\epsilon=0$). 

However, if we know more about the structure of the channel, we can give significantly better bounds. We will do this here using Equation~(\ref{Eq:Dp-Eg-bound}). 
\begin{lemma}\label{Lem:Dp-better-bounds}
Say, $\rho\sim\sigma$ if $\frac12 \|\rho-\sigma\|_1 \leq\kappa$, then $\cD_p$ is $(\epsilon, \delta)$-differentially private with
\begin{align}
    \delta=\max\{0,(1-e^{\epsilon})\frac{p}{D} + (1-p)\kappa \}. 
\end{align}
For an algorithm $\cA$ of the form in Equation~\eqref{Eq:LayerA} with all $\cN_i=\cD_{p_i}$ and $\rho\sim\sigma$ if $\frac12 \|\rho-\sigma\|_1 \leq\kappa$, then $\cA$ is $(\epsilon, \delta)$-differentially private with
\begin{align}
    \delta=\max\{0,(1-e^{\epsilon})\frac{p_\star}{D} + (1-p_\star)\kappa \}, \label{Eq:Dp-better-bound}
\end{align}
where $p_\star=1-\prod_i(1-p_i)$.
\end{lemma}
\begin{proof}
The first statement follows directly from Equation~(\ref{Eq:Dp-Eg-bound}). The second statement can be proven by induction. Recall that $n$ is the number of steps in the algorithm, i.e. the number of depolarizing layers. For $n=1$ the statement follows from~(\ref{Eq:Dp-Eg-bound}). Now, assuming the results holds for $n-1$ layers, for the $n$-th layer we have
\begin{align*}
    &(1-\gamma)\frac{p_\star}{D} + (1-p_\star)E_\gamma(\cD_p(\rho)\|\cD_p(\sigma)) \\
    &\leq (1-\gamma)\frac{p_\star}{D} \\
    &\quad+ (1-p_\star)\left((1-\gamma)\frac{p_n}{D} + (1-p_n)E_\gamma(\rho\|\sigma)\right) \\
    &= (1-\gamma)\frac{1 - (1-p_\star)(1-p_n)}{D} \\
    &\quad+ (1-p_\star)(1-p_n)E_\gamma(\rho\|\sigma) 
\end{align*}
from which the claim follows directly. 
\end{proof}

\begin{figure*}
    \centering
    \begin{minipage}{0.47\textwidth}
        \centering
        \includegraphics[width=1\textwidth]{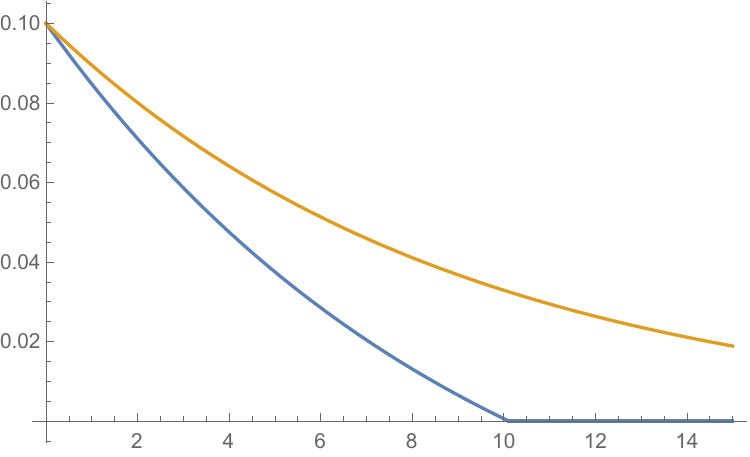} 
    \end{minipage}
    \hfill
    \begin{minipage}{0.47\textwidth}
        \centering
        \includegraphics[width=1\textwidth]{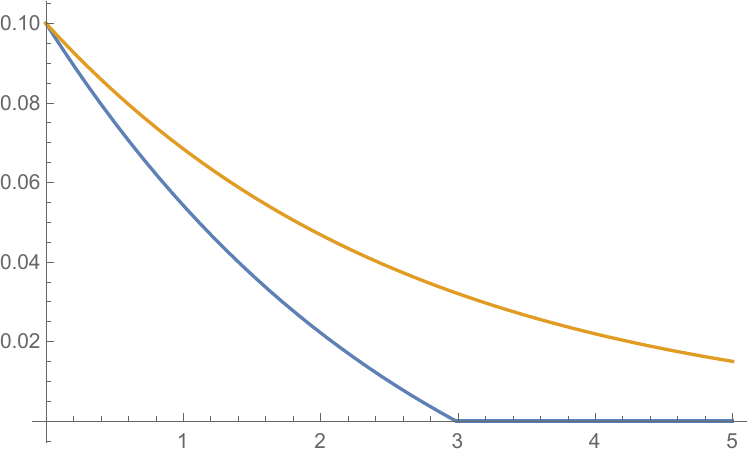} %
    \end{minipage}
\caption{\label{Fig:Depolarizing-bounds} Comparison of bounds on $\delta$ for $\epsilon=0.1$, $\kappa=0.1$ and $D=2$ plotted against $n$. The top function is the contraction coefficient bound in Equation~\eqref{Eq:delta-decay} and below the improved bound in Equation~\eqref{Eq:Dp-better-bound}. Left: $p=0.1$. Right: $p=0.3$.}
\end{figure*} 

The above bound is compared to the direct contraction coefficient bound from Equation~\eqref{Eq:delta-decay} in Figure~\ref{Fig:Depolarizing-bounds}. 
As can be seen, the improvement is significant. While contraction coefficients give an exponential decay of $\delta$ with $n$, the bound in Equation~\eqref{Eq:Dp-better-bound} shows that $\delta=0$ is sufficient already for $n\geq 3$ if $p=0.3$ and $n\geq 10$ if $p=0.1$. We remark that for $\epsilon=0$ the two bounds give the same result, i.e. in case of the trace distance Lemma~\ref{Lem:Dp-better-bounds} does not lead to an improvement over the contraction bound. In fact, for the trace distance this simple bound is tight in some cases, because
\begin{align}
    \| \cD_p(\rho)-\cD_p(\sigma)\|_1 &= \| (1-p)(\rho-\sigma)\|_1 \\
    &= (1-p) \|\rho-\sigma\|_1\,,  \label{Eq:Dep-Equality}
\end{align}
implying that the decrease in trace distance is always exactly $(1-p)$ for a layer of global depolarizing noise. This easily extends to algorithms $\cA$ of the form in Equation~\eqref{Eq:LayerA} when all computational layers $\mathcal{C}_i$ are unitary, i.e. don't introduce any noise themselves. 
These observations lead to a provable separation in $n$ between good privacy and useless algorithm outputs. 

Alternatively we can also state a similar result determining a bound on $\epsilon$ as a simple corollary. 
\begin{corollary}\label{Cor:Dp-eps}
Say, $\rho\sim\sigma$ if $\frac12 \|\rho-\sigma\|_1 \leq\kappa$. For a fixed $\delta\geq 0$, $\cD_p$ is $(\epsilon, \delta)$-differentially private with
\begin{align}
    \epsilon\geq\max\{0,\log\left(\frac{D}{p}((1-p)\kappa-\delta)+1\right)\}.
\end{align}
For an algorithm $\cA$ of the form in Equation~\eqref{Eq:LayerA}, for a fixed $\delta\geq 0$, $\cA$ is $(\epsilon, \delta)$-differentially private with
\begin{align}
    \epsilon\geq\max\{0,\log\left(\frac{D}{p_\star}((1-p_\star)\kappa-\delta)+1\right)\}. \label{Eq:eps-low-dp-gl}
\end{align}
where $p_\star=1-\prod_i(1-p_i)$.
\end{corollary}
This generalizes some results in the literature, namely, for $\delta=0$ the first statement reduces to~\cite[Theorem 3]{zhou2017differential} and the second to~\cite[Lemma 2]{du2021quantum}. We present examples of the above bound in Figure~\ref{Fig:eps-bounds-global-dep}, showcasing the trade-off between $\epsilon$ and $\delta$ and the decay of $\epsilon$ with growing $n$. 

\begin{figure*}
    \centering
    \begin{minipage}{0.47\textwidth}
        \centering
        \includegraphics[width=1\textwidth]{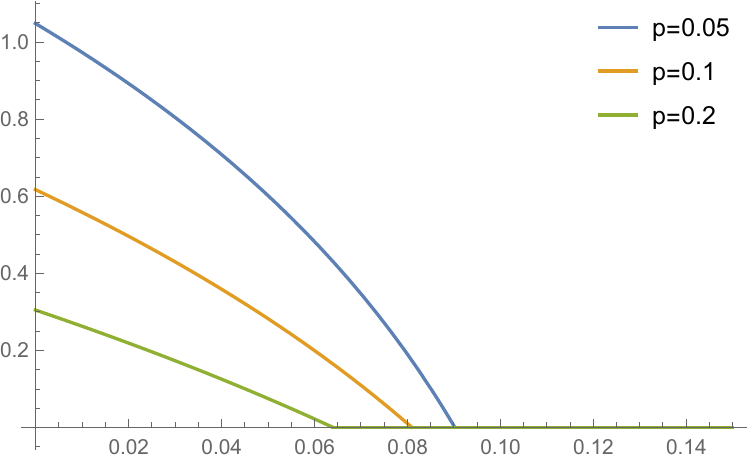} %
    \end{minipage}\hfill
    \begin{minipage}{0.47\textwidth}
        \centering
        \includegraphics[width=1\textwidth]{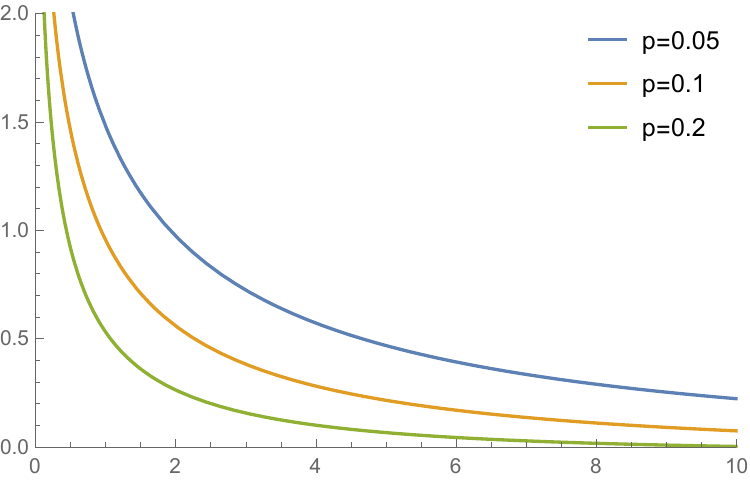} %
    \end{minipage}
\caption{\label{Fig:eps-bounds-global-dep} Bound on $\epsilon$ in Equation~\eqref{Eq:eps-low-dp-gl} for $D=2$ and $\kappa=0.1$ for different values of $p_i=p$. Left: plotted for $n=2$ against $\delta$. Right: plotted for $\delta=0.01$ against $n$. }
\end{figure*}

\subsection{Local depolarizing noise}

Previously we have seen how quantum algorithms are affected by global depolarizing noise. However, in a quantum computing device we would rather expect each qubit to be affected by local noise. In this section we will discus depolarizing noise of the form $\cD_p^{\otimes k}$ where $k$ is the number of qubits a quantum algorithm acts on at any given layer. To investigate this setting we first need an equivalent of Lemma~\ref{Lem:HS-Dp-bounds} for local depolarizing noise.
\begin{lemma}\label{Lem:HS-Dp-local-bounds}
For $0\leq p\leq 1$ and $\gamma\geq1$ we have
\begin{align}
        &E_\gamma(\cD_p^{\otimes k}(\rho)\|\cD_p^{\otimes k}(\sigma)) \\
        &\leq \max\{0,(1-\gamma)\frac{p^k}{D^k} + (1-p^k)E_\gamma(\rho\|\sigma)\} \label{Eq:Dp-Eg-bound-local}
\end{align}
and 
\begin{align}
    \eta_\gamma(\cD_p^{\otimes k})\leq \max\{0,(1-\gamma)\frac{p^k}{D^k} + (1-p^k)\}. \label{Eq:Dp-cc-local}
\end{align}
\end{lemma}
\begin{proof}
The proof is similar to Lemma~\ref{Lem:HS-Dp-bounds} but requires one more main ingredient. Note that we can always write local depolarizing noise as
\begin{align*}
    \cD_p^{\otimes k}(\rho) = p^k \frac{\Id^{\otimes k}}{D^k} + (1-p^k) \cM(\rho), 
\end{align*}
where $\cM$ is some CPTP map. This can be checked by direct calculation. With this we have
\begin{align*}
    &E_\gamma(\cD_p(\rho)\|\cD_p(\sigma)) \\
    &= \tr((1-\gamma)p^k\frac{\Id^{\otimes k}}{D^k} + (1-p^k)\cM((\rho-\gamma\sigma)))^+ \\
    &= \tr P^+((1-\gamma)p^k\frac{\Id^{\otimes k}}{D^k} + (1-p^k)\cM((\rho-\gamma\sigma))), 
\end{align*}
where $P^+$ is the projector onto the positive subspace of $((1-\gamma)p^k\frac{\Id^{\otimes k}}{D^k} + (1-p^k)\cM((\rho-\gamma\sigma)))$. Observe that 
\begin{align*}
    E_\gamma(\cD_p^{\otimes k}(\rho)\|\cD_p^{\otimes k}(\sigma))>0 \quad\Rightarrow\quad \tr P^+\geq 1.
\end{align*}
Considering this case we get
\begin{align*}
    &E_\gamma(\cD_p(\rho)\|\cD_p(\sigma)) \\
    &= (1-\gamma)\frac{p^k}{D^k}\tr P^+ + (1-p^k)(\tr P^+(\cM(\rho-\gamma\sigma))) \\
    &\leq (1-\gamma)\frac{p^k}{D^k} + (1-p^k)E_\gamma(\cM(\rho)\|\cM(\sigma)) \\
    &\leq (1-\gamma)\frac{p^k}{D^k} + (1-p^k)E_\gamma(\rho\|\sigma) \\
    &\leq (1-\gamma)\frac{p^k}{D^k} + (1-p^k). 
\end{align*}
Note that for sufficiently large $\gamma$ the upper bound could become negative, but one can easily check that in this case $E_\gamma(\cD_p^{\otimes k}(\rho)\|\cD_p^{\otimes k}(\sigma))=0$ implying that we are in the other case. 
\end{proof}
Note that in this case we provide only an upper bound on the contraction coefficient and determining its exact value remains open. 
Nevertheless, with the above tool at hand, we can easily generalize Lemma~\ref{Lem:Dp-better-bounds}.
\begin{lemma}\label{Lem:Dp-better-local-bounds}
Say, $\rho\sim\sigma$ if $\frac12 \|\rho-\sigma\|_1 \leq\kappa$, then $\cD_p^{\otimes k}$ is $(\epsilon, \delta)$-differentially private with
\begin{align}
    \delta=\max\{0,(1-e^{\epsilon})\frac{p^k}{D^k} + (1-p^k)\kappa \}. 
\end{align}
For an algorithm $\cA$ of the form in Equation~\eqref{Eq:LayerA} with all $\cN_i=\cD_{p}^{\otimes k}$ and $\rho\sim\sigma$ if $\frac12 \|\rho-\sigma\|_1 \leq\kappa$, then $\cA$ is $(\epsilon, \delta)$-differentially private with
\begin{align}
    \delta=\max\{0,(1-e^{\epsilon})\frac{p_\star}{D^k} + (1-p_\star)\kappa \}, \label{Eq:Dp-better-bound-local}
\end{align}
where $p_\star=1-(1-p^k)^n$.
\end{lemma}
\begin{proof}
The proof is identical to that of Lemma~\ref{Lem:Dp-better-bounds}. 
\end{proof}

\begin{figure*}
    \centering
    \begin{minipage}{0.47\textwidth}
        \centering
        \includegraphics[width=1\textwidth]{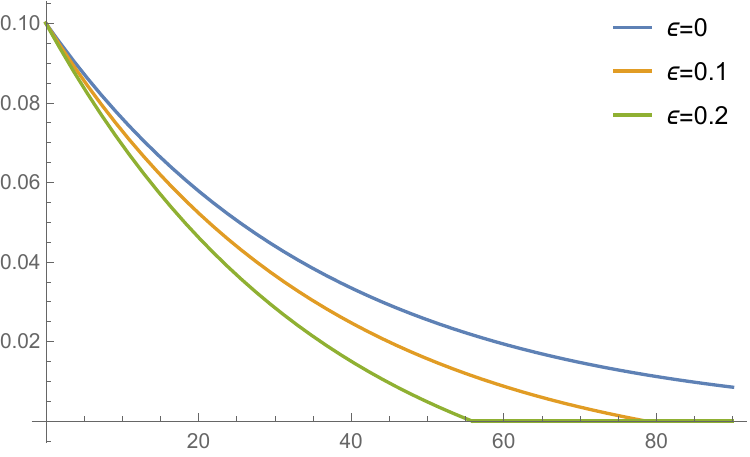} %
    \end{minipage}\hfill
    \begin{minipage}{0.47\textwidth}
        \centering
        \includegraphics[width=1\textwidth]{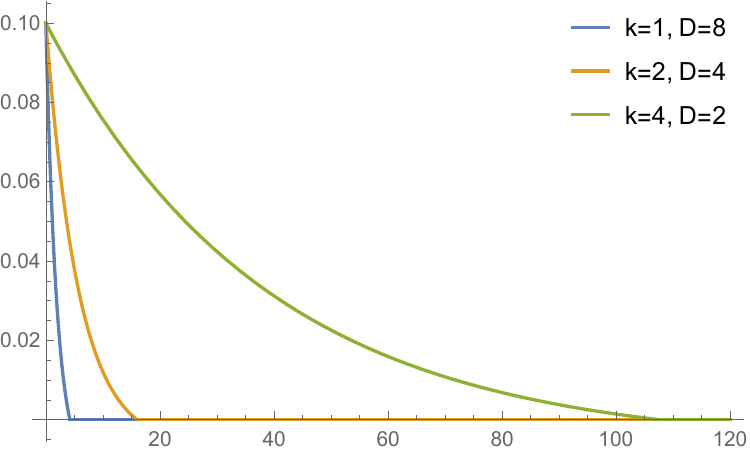} %
    \end{minipage}
\caption{\label{Fig:Depolarizing-local} Left: The bound in Equation~\eqref{Eq:Dp-better-bound-local} on $\delta$ for $p=0.3, k=3, D=2, \kappa=0.1$ plotted against $n$ for different $\epsilon$. Right: The same bound for $p=0.4, \epsilon=0.1, \kappa=0.1$ for different values of $k$ and $D$, such that the total dimension at each layer is always $8$, plotted against $n$.}
\end{figure*}

The bounds on $\delta$ are illustrated in Figure~\ref{Fig:Depolarizing-local}. On the left we see that for values $\epsilon>0$, $\delta$ is eventually going to reach $0$, while for $\epsilon=0$ i.e. the case of the trace distance, $\delta$ decays exponentially but always stays strictly positive. This is the same as for global depolarizing noise. On the right, we compare local and global depolarizing noise. For that we fix a total dimension of $8$ at each layer and consider different dimension of the depolarizing noise, from a single global depolarizing channel to $4$ local qubit depolarizing channels. From the plot it becomes evident that our bounds guarantee much faster decay of $\delta$ for global than for local noise.  

Finally, we can also adapt Corollary~\ref{Cor:Dp-eps} to local noise.
\begin{corollary}\label{Cor:Dp-eps-local}
Say, $\rho\sim\sigma$ if $\frac12 \|\rho-\sigma\|_1 \leq\kappa$. For a fixed $\delta\geq 0$, $\cD_p$ is $(\epsilon, \delta)$-differentially private with
\begin{align}
    \epsilon\geq\max\{0,\log\left(\frac{D^k}{p^k}((1-p^k)\kappa-\delta)+1\right)\}.
\end{align}
For an algorithm $\cA$ of the form in Equation~\eqref{Eq:LayerA} with all $\cN_i=\cD_{p}^k$. For a fixed $\delta\geq 0$, $\cA$ is $(\epsilon, \delta)$-differentially private with
\begin{align}
    \epsilon\geq\max\{0,\log\left(\frac{D^{k}}{p_\star}((1-p_\star)\kappa-\delta)+1\right)\}.
\end{align}
where $p_\star=1-(1-p^k)^n$.
\end{corollary}

We end this section by providing a lower bound on the trace distance for local depolarizing noise that will demonstrate the substantial difference between decaying computational accuracy and good differential privacy. The bound takes a similar role to the observation around Equation~\eqref{Eq:Dep-Equality} in the previous section. 

\begin{proposition}
For any two states $\rho,\sigma$ and local depolarizing parameter $0\le p\le 1$ such that $p<\frac{1}{2}$, 
\begin{align*}
    \|\rho-\sigma\|_1\le \Big(\frac{1}{1-2p}\Big)^k\,\|\cD_p^{\otimes k}(\rho-\sigma)\|_1\,.
\end{align*}
Therefore, for the noisy circuit $\mathcal{A}$ where all the gates $\mathcal{C}_i$ are chosen to be unitary, and where $\cN_i=\mathcal{D}_p^{\otimes k}$, we have that
\begin{align}
    E_1(\cA(\rho)\|\cA(\sigma)) \ge \,\frac12\Big({1-2p}\Big)^{kn}\,\|\rho-\sigma\|_1\,. \label{Eq:DeltaLower}
\end{align}
\end{proposition}

\begin{proof}
We start by inverting the depolarizing noise acting on qubit $i$ as:
\begin{align*}
    \rho=\frac{\cD_p(\rho)-\frac{p}{D}\tr_i(\rho)\otimes \Id_i}{1-p}\,
    \end{align*}
    Hence, we have
    \begin{align*}
\|\rho-\sigma\|_1 &=\frac{1}{1-p}\,\|\cD_p(\rho-\sigma)-\frac{p}{D}\,\tr_i(\rho-\sigma)\otimes \Id_i\|_1 \\
&\le \frac{1}{1-p}\,\|\cD_p(\rho-\sigma)\|_1+\frac{p}{(1-p)}\,\|\rho-\sigma\|_1\,.
\end{align*}
Therefore, whenever $p<\frac{1}{2}$, we have that
\begin{align*}
    \|\rho-\sigma\|_1\le \frac{1}{1-2p}\,\|\cD_p(\rho-\sigma)\|_1\,.
\end{align*}
The result arises from repeating the above step for all the qubits.
\end{proof}

We compare the above lower bound to our previous upper bounds in Figure~\ref{Fig:DeltaLower}. It can be seen that there is a clear separation between the worst case decay of the trace distance and the depth required to reach differential privacy. It should however be noted that this bound seems to be  useful only for small $k$. 

\begin{figure*}
    \centering
    \begin{minipage}{0.48\textwidth}
        \centering
        \includegraphics[width=0.99\textwidth]{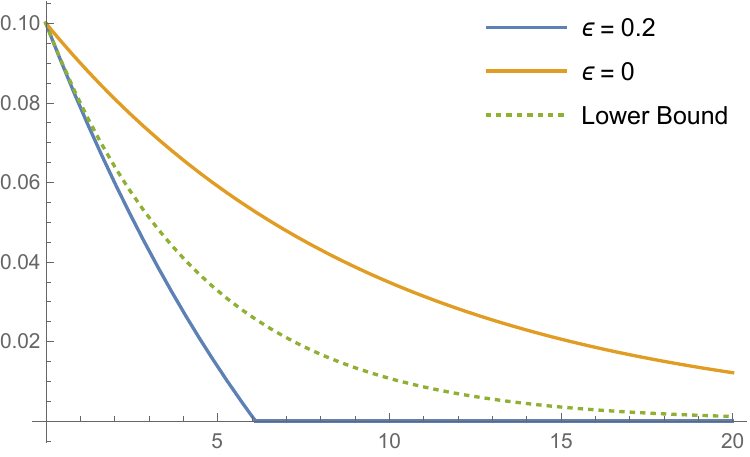} %
    \end{minipage}\hfill
    \begin{minipage}{0.48\textwidth}
        \centering
        \includegraphics[width=0.99\textwidth]{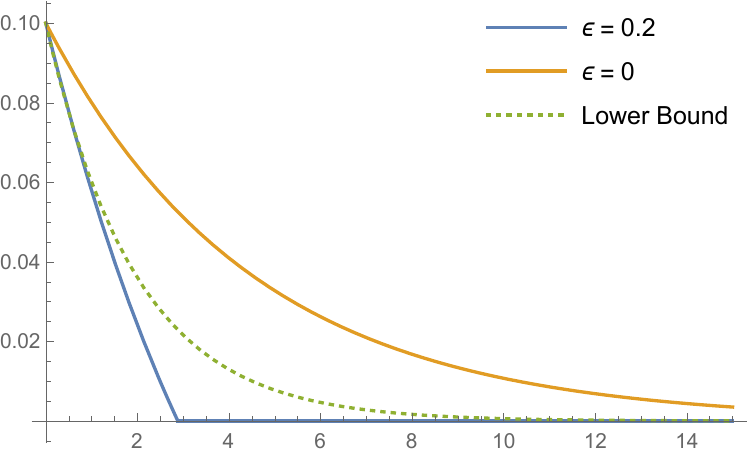} %
    \end{minipage}
\caption{\label{Fig:DeltaLower} Comparison of the bound in Equation~\eqref{Eq:Dp-better-bound-local}, solid lines for different $\epsilon$, with the lower bound in Equation~\eqref{Eq:DeltaLower}, dotted line. Plotted for $D=2$, $k=1$ and $\frac12 \|\rho-\sigma\|_1 = 0.1$ against $n$, with $p=0.1$ on the left and $p=0.2$ on the right.}
\end{figure*}

The previous results suggest that the local noise channels do not contract too fast at small enough noise. Here, we prove that the trace distance, and hence also any hockey stick divergence, will converge exponentially fast in the number of layers as soon as a critical local noise is attained. Our result can be compared to earlier upper bounds on the noise threshold for fault-tolerant quantum computing as in \cite{kempe2008upper,razborov2003upper}. However, as we will see, our bound comes as a simple corollary of basic properties of a recently introduced quantum Wasserstein distance in \cite{WassersteindP}. Although the argument can be generalized to other Pauli noisy channels (see e.g. \cite{gao2021ricci}), we restrict once again our analysis to the simple depolarizing channel. Here, given a quantum channel $\Phi$ acting on $k$ qubits, we define its light-cone as follows: first, for any qubit $i$, we denote by $I_i$ the minimal subset of qubits such that $\tr_{I_i}(\Phi(\rho))=\tr_{I_i}(\Phi(\sigma))$ for any two $k$-qubit states $\rho$ and $\sigma $ such that $\tr_i(\rho)=\tr_i(\sigma)$. Then, the light-cone of $\Phi$ is defined as
\begin{align*}
|I|:=\max_{i\in [k] } \,|I_i|\,.    
\end{align*}

\begin{proposition}
Given the noisy circuit $\mathcal{A}$ in \Cref{Eq:LayerA} with $n$ layers, $k$ qubits and local depolarizing noise of parameter $0\le p\le 1$, we assume that each layer of the circuit is a quantum channel of light-cone $|I|$. Then, we have that for any two input states $\rho,\sigma$
\begin{align*}
    \|\mathcal{A}(\rho)-\cA(\sigma)\|_1\le \,2k\,\big(2|I|\,(1-p)\big)^n\,.
\end{align*}
In other words, the trace distance between any two output states vanishes in logarithmic depth as soon as $p$ satisfies $2|I|(1-p)<1$.
\end{proposition}

\begin{proof}
We will use the Wasserstein distance introduced in \cite{WassersteindP} as follows:
\begin{align*}
    W_1(\rho,\sigma):=\sup_{\|O\|_L\le 1}\,\tr(O(\rho-\sigma))\,, 
\end{align*}
where 
\begin{align*}
    \|H\|_L:=2\,\max_{i\in[k]}\,\min_{H^{(i)}}\,\|H-H^{(i)}\|_\infty\,,
\end{align*} 
where the mininmum is over all self-adjoint operators $H^{(i)}$ which do not act on qubit $i$. Next, we have
\begin{align*}
    &\frac{1}{2}\,\|\mathcal{A}(\rho-\sigma)\|_1\le W_1(\mathcal{A}(\rho),\mathcal{A}(\sigma)) \\
    &\le (2|I|(1-p))^n W_1(\rho,\sigma) \\
    &\le \frac{k}{2}\,(2|I|(1-p))^n\,\|\rho-\sigma\|_1
\end{align*}
Above, the first and last inequalities follow from \cite[Proposition 2]{WassersteindP}, whereas the second inequality comes from an alternating use of \cite[Propositions 12 and 13]{WassersteindP}. The result follows. 
\end{proof}
Compared to our earlier results this bounds dependence on $k$ is weaker and we recover the $(1-p)^n$ scaling, previously seen for the global depolarizing noise, whenever $p$ is large enough to overpower the, possibly error-correcting, properties of the computational layers.

\subsection{Arbitrary local qubit noise channels}
In this section we will give a bound for the contraction of arbitrary local qubit channels based on recent work in~\cite{fawzi2022space}. In particular we will see that any non-unital local qubit noise leads to good differential privacy eventually. The main result is the following lemma that generalizes~\cite[Lemma 7]{fawzi2022space} by combining their proof with our generalized Fuchs-van-de-Graaf inequality. 
\begin{lemma}
Let $\cT=\cN^{\otimes k}$ and $\cN$ some qubit channel, then
\begin{align}
    &\eta_\gamma(\cT) \nonumber\\
    &\leq \frac12\sqrt{(1+\gamma)^2 - 4\gamma \left(\frac{\lambda_{\min}(C_{\cN^\dagger\circ\cN})}{4}\right)^k} + \frac{(1-\gamma)}{2}, 
\end{align}
where $C_{\cN^\dagger\circ\cN}$ is the Choi matrix of $\cN^\dagger\circ\cN$ and $\lambda_{\min}(C_{\cN^\dagger\circ\cN})$ its smallest eigenvalue. If $\cN$ is also non-unital then
\begin{align}
    \lambda_{\min}(C_{\cN^\dagger\circ\cN})>0.
\end{align}
\end{lemma}
\begin{proof}
We begin by applying the generalized Fuchs-van-de-Graaf inequality from Lemma~\ref{Lem:HS-FvdG} to the simplified expression of the contraction coefficient in Theorem~\ref{Thm:ConCoef-easy} as previously stated in Equation~\eqref{Eq:eta-FvdG}. 
We then apply a bound on the fidelity shown in~\cite{fawzi2022space} that states 
\begin{align*}
    F(\cN^{\otimes k}(\Psi),\cN^{\otimes k}(\Phi)) \geq \left(\frac{\lambda_{\min}(C_{\cN^\dagger\circ\cN})}{4}\right)^k,
\end{align*}
see Appendix~\ref{Sec:Lemmas} for the details. The final statement about non-unital qubit channels was also argued in~\cite{fawzi2022space}.  
\end{proof}

This bound applies directly to differential privacy via Corollary~\ref{Cor:delta-decay}. Let $\cA$ be of the form in Equation~\eqref{Eq:LayerA} with all $\cN_i=\cN^{\otimes k}$ identical and $\rho\sim\sigma$ if $\frac12 \|\rho-\sigma\|_1 \leq\kappa$. Let $\lambda_{\min}(C_{\cN^\dagger\circ\cN})=\lambda$, then $\cA$ is $(\epsilon, \delta)$-differentially private with
\begin{align}
\delta = \left(\frac12\sqrt{(1+e^\epsilon)^2 - 4e^\epsilon \left(\frac{\lambda}{4}\right)^k} + \frac{(1-e^\epsilon)}{2}\right)^n \kappa. \label{Eq:delta-fid-bound}
\end{align}

\begin{figure*}
    \centering
    \begin{minipage}{0.47\textwidth}
        \centering
        \includegraphics[width=1\textwidth]{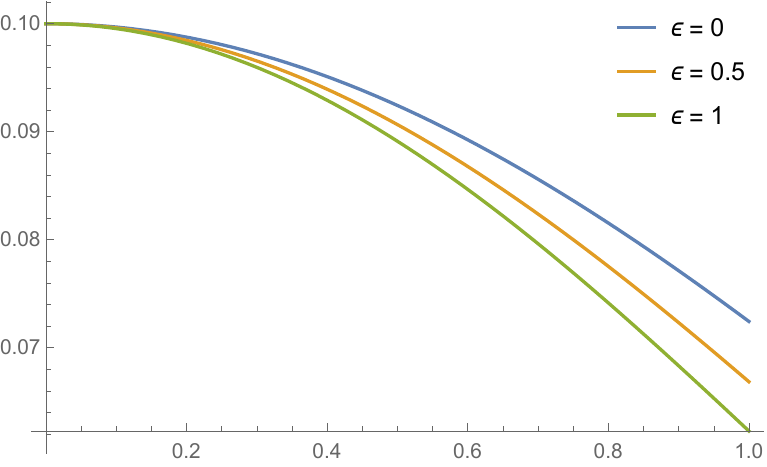} %
    \end{minipage}\hfill
    \begin{minipage}{0.47\textwidth}
        \centering
        \includegraphics[width=1\textwidth]{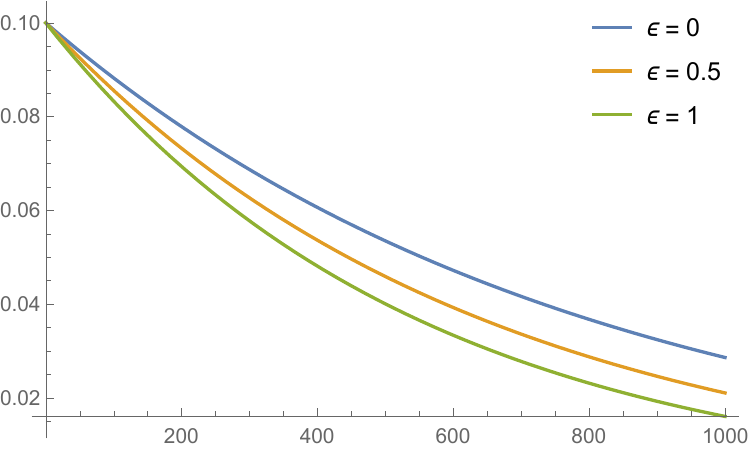} %
    \end{minipage}
\caption{\label{Fig:local-fid-bound} Guaranteed $\delta$ for some local noise channel based on Equation~\eqref{Eq:delta-fid-bound} with $\kappa=0.1$, $k=2$ and different values of $\epsilon$. Left: plotted for $n=10$ against $\lambda$. Right: plotted for $\lambda=0.2$ against $n$.}
\end{figure*}

We observe that for growing $\epsilon$ we get smaller $\delta$, a natural trade-off between the two security parameters, also implying that $\delta$ always decays faster than the trace distance ($\epsilon=0$). Also for any channel with $\lambda>0$, which is the case for every non-unital channel as shown above, good differential privacy will be eventually reached for large enough $n$. We give some numerical examples in Figure~\ref{Fig:local-fid-bound}.

\section{Extensions} \label{Sec:Extensions}

\subsection{Local quantum differential privacy}

Local differential privacy (LDP) was defined in the classical setting for the scenario in which a database is collected from many clients and each of them demands differential privacy to hold for their individual contribution. In this case the client applies an algorithm $\cA$ to mask their contribution to the database and the priority is not to make neighboring states look similar but to hide the general information they are sending. The definition of $(\epsilon,\delta)$-LDP therefore coincides with that $(\epsilon,\delta)$-DP but with a more general set of possible input states. 
In the extreme setting one could even consider the set of all possible input states, implying that $\cA$ is $(\epsilon,\delta)$-LDP if
\begin{align}\label{Eq:LDP-def}
    \sup_{\rho,\sigma} E_{e^\epsilon}(\cA(\rho) \| \cA(\sigma) ) \leq \delta\,, 
\end{align}
based on Lemma~\ref{Lem:DP-divergence}. Clearly this is a much stronger requirement than what is required for $(\epsilon,\delta)$-DP. In fact from the properties of $E_\gamma$ one can see that this condition is restrictive enough to imply a bound on the trace distance contraction coefficient. This generalizes a classical result in~\cite{asoodeh2021local}. 
\begin{corollary}\label{Cor:bound-by-LDP}
Let $\cA$ be $(\epsilon,\delta)$-LDP and $\varphi(\epsilon,\delta)=1-e^{-\epsilon}(1-\delta)$. Then 
\begin{align}
    \eta_1(\cA) \leq \varphi(\epsilon,\delta). 
\end{align}
\end{corollary}
\begin{proof}
Clearly, by definition of LDP and Theorem~\ref{Thm:ConCoef-easy} we have 
\begin{align*}
    \eta_{e^\epsilon}(\cA) \leq \delta. 
\end{align*}
Now the corollary follows because from Lemma~\ref{Lem:gamma-trace-distance} we get
\begin{align*}
    \eta_1(\cA)\leq 1 - \frac{1-\eta_\gamma(\cA)}{\gamma}, 
\end{align*}
which together with the definition of $(\epsilon,\delta)$-LDP concludes the proof.
\end{proof}
In particular, this implies that requiring $\cA$ to be $(\epsilon,\delta)$-LDP can strongly limit the usefulness of the output states. In particular applying $\cA$ iteratively will lead to very strong privacy guarantees but also make the output states indistinguishable and therefore useless for further computations, compare Proposition~\ref{Thm:Contraction}. 

Note that the above argument also works the same with somewhat less restrictive definitions.
In particular, we get the following equivalence, which is similar to an observation in~\cite{aaronson2019gentle}.
\begin{corollary}
$\cA$ being $(\epsilon,\delta)$-LDP with respect to the set of all states is equivalent to $\cA$ being $(\epsilon,\delta)$-LDP with respect to the set of all pure states.
\end{corollary}
\begin{proof}
This follows directly from the convexity of $E_\gamma$. 
\end{proof}
In principle, Corollary~\ref{Cor:bound-by-LDP} also holds for $(\epsilon,\delta)$-LDP with respect to any set of states that includes all orthogonal pure states. For the minimal such set, $(\epsilon,\delta)$-LDP would not just imply but indeed be equivalent to $\eta_{e^\epsilon}(\cA) \leq \delta$ as observed in the classical setting~\cite{asoodeh2021local}.

\subsection{\Renyi quantum differential privacy}\label{Sec:Renyi}

In many practical settings $\epsilon$-differential privacy can be too strong of a criteria. On the other hand, $(\epsilon,\delta)$-differential privacy allows for rare events that leak a significant amount of information and its composition theorem requires adding the $\delta$'s of each algorithm. As an intermediate privacy requirement Mironov proposed \Renyi differential privacy~\cite{mironov2017renyi} and proved that it has several desirable properties. As we have seen in the previous sections, $\epsilon$-differential privacy can be cast in terms of the $\max$-relative entropy. Essentially, $(\epsilon,\alpha)$-\Renyi differential privacy is defined by replacing the $\max$-relative entropy by the general $\alpha$-\Renyi relative entropy. In this section we will propose a quantum extension of this concept. 

The obvious generalization of the classical definition is to go a similar route to the original quantum differential privacy definition and define measurement outcomes of a POVM $\{M_x\}$ as $p(x)=\tr M_x \cA(\rho)$ and $q(x)=\tr M_x \cA(\sigma)$ and require the classical definition to hold. This can be cast in terms of the measured \Renyi relative entropy,
\begin{align}
   \sup_{\rho\sim\sigma} D_{M,\alpha}(\cA(\rho)\|\cA(\sigma)) = \sup_{\rho\sim\sigma} \sup_{\{M_x\}} D_{\alpha}(p\|q) \leq \epsilon.  \label{Eq:sc-RDP}
\end{align}
This coincides with $\epsilon$-differential privacy for $\alpha\rightarrow\infty$, however for other values of $\alpha$ there is no known closed (measurement independent) formula for the measured \Renyi relative entropy and the definition remains classical at its heart. This also comes with a concrete disadvantage when considering composition bounds, namely that the measured \Renyi relative entropy is not generally subadditive\footnote{This follows e.g. because the measured \Renyi relative entropy is strictly smaller than the sandwiched \Renyi relative entropy, see e.g.~\cite[Theorem 6]{berta2017variational}, however they are equal under regularization, see Equation~\eqref{Eqn:sandwichedRegularized}}. That means there exist examples where 
\begin{align}
   &D_{M,\alpha}(\cA_1\otimes\cA_2(\rho_1\otimes\rho_2)\|\cA_1\otimes\cA_2(\sigma_1\otimes\sigma_2)) \nonumber\\
   &>  D_{M,\alpha}(\cA_1(\rho_1)\|\cA_1(\sigma_1)) + D_{M,\alpha}(\cA_2(\rho_2)\|\cA_2(\sigma_2))
\end{align}
implying that, using this quantity, $(\epsilon,\alpha)$-\Renyi differential privacy of $\cA_1$ and $\cA_2$ generally does not imply $(2\epsilon,\alpha)$-\Renyi differential privacy for $\cA_1\otimes\cA_2$. 
In this section, we want to propose a fully quantum definition of $(\epsilon,\alpha)$-\Renyi differential privacy that avoids this problem. 

While in the classical setting there is a uniquely defined \Renyi relative entropy, in the quantum setting, due to its non-commutative nature, there is an arbitrary number of generalizations. Here, we will not fix a particular definition but simply consider an arbitrary family of \Renyi relative entropies $\DD_\alpha$ as defined in~\cite{tomamichel2015quantum} based on a number of requirements that the quantity needs to fulfill. For completeness a list of the properties can be found in Appendix~\ref{Sec:App-Renyi} and we refer to~\cite{tomamichel2015quantum} for details. Here we only note that the list includes several typical properties such as data-processing, additivity and unitary invariance. We further note that the commonly used quantum generalizations such as the Petz-\Renyi divergence $D_\alpha$~\cite{petz1986quasi}, the sandwiched \Renyi divergence $\tilde D_\alpha$~\cite{Muller-Lennert2013,Wilde2014strong} and the geometric \Renyi divergence $\hat D$~\cite{matsumoto2015new} are special instances of $\DD_\alpha$ for the range of $\alpha$ in which the mentioned properties hold. Formally, we define quantum \Renyi differential privacy as follows.
\begin{definition}
We call a quantum channel $\cA$ $(\epsilon,\alpha)$-\Renyi differentially private if 
\begin{align}
    \sup_{\rho\sim\sigma} \DD_\alpha(\cA(\rho)\|\cA(\sigma)) \leq \epsilon. 
\end{align}
\end{definition}
This generally leaves us with a lot of freedom to pick our favourite \Renyi relative entropy. Note however that those which include the limit $\alpha\rightarrow\infty$ all have the particular feature that in this limit the above definition includes $\epsilon$-differential privacy as a special case. The sandwiched \Renyi relative entropy is an example of such a family. We now state the first property of our definition. 
\begin{lemma}
If $\cA$ is $\epsilon$-differentially private then it is $(\epsilon,\alpha)$-\Renyi differentially private. 
\end{lemma}
\begin{proof}
The claim is a direct consequence of
\begin{align*}
    \DD_\alpha(\cA(\rho)\|\cA(\sigma)) \leq D_{\max}(\cA(\rho)\|\cA(\sigma)),
\end{align*}
which can be seen as follows. We have the following chain of arguments,
\begin{align*}
    \DD_\alpha(\rho\|\sigma) \leq \hat D_\alpha(\rho\|\sigma) \leq \hat D_\infty(\rho\|\sigma)=D_{\max}(\rho\|\sigma),
\end{align*}
where the first inequality is~\cite[Equation (4.34)]{tomamichel2015quantum}, the second is monotonicity of $\hat D_\alpha$ in $\alpha$ and the equality is~\cite[Equation (4.36)]{tomamichel2015quantum}.  
\end{proof}
The classical \Renyi differential privacy is furthermore known to be stronger than $(\epsilon,\delta)$-differential privacy. We now show that the same holds for our quantum definition. 
\begin{lemma}
If $\cA$ is $(\epsilon,\alpha)$-\Renyi differentially private then it is $(\epsilon + \frac{g(\delta)}{\alpha-1},\delta)$-differentially private with $g(\epsilon)=-\log(1-\sqrt{1-\epsilon^2})$.  
\end{lemma}
\begin{proof}
The main ingredients are~\cite[Proposition 6.22]{tomamichel2015quantum} as stated in the appendix and an auxiliary lemma relating $D^\delta_{\max}$ with a similar quantity more commonly found in the literature stated in Lemma~\ref{Lem:aux-Dmax}. From it we have 
\begin{align*}
    D^\delta_{\max}(\rho\|\sigma) \leq \DD_\alpha(\rho\|\sigma) + \frac{g(\delta)}{\alpha-1} \leq \epsilon + \frac{g(\delta)}{\alpha-1}, 
\end{align*}
where the last inequality is by assumption. The proof follows from Lemma \ref{Lem:DP-divergence}. 
\end{proof}
This can be relaxed by noting that $g(\delta)\leq\log\frac2{\delta^2}$, which brings it closer to the classical equivalent~\cite[Proposition 3]{mironov2017renyi}. 

Additionally, we can easily show several desirable properties. 
\begin{corollary}
The following properties hold.
\begin{itemize}
    \item{(Post-processing)} Let $\cA$ be $(\epsilon,\alpha)$-\Renyi differentially private and $\cN$ be an arbitrary quantum channel, then $\cN\circ\cA$ is also $(\epsilon,\alpha)$-\Renyi differentially private. 
    \item{(Parallel composition)} Let $\cA_1$ be $(\epsilon_1,\alpha)$-\Renyi differentially private and $\cA_2$ be $(\epsilon_2,\alpha)$-\Renyi differentially private. Define that $\rho_1\otimes\rho_2\sim\sigma_1\otimes\sigma_2$ if $\rho_1\sim\sigma_1$ and $\rho_2\sim\sigma_2$. Then $\cA_1\otimes\cA_2$ is $(\epsilon_1+\epsilon_2,\alpha)$-\Renyi differentially private on such product states. 
\end{itemize}
\end{corollary}
\begin{proof}
The first result follows by data-processing and the second by additivity of $\DD_\alpha$. 
\end{proof}
Finally, we remark that our general quantum definition of \Renyi differential privacy implies \Renyi differential privacy in the semi-classical definition in Equation~\eqref{Eq:sc-RDP} because
\begin{align}
    D_{M,\alpha}(\cA(\rho)\|\cA(\sigma)) \leq \DD_\alpha(\cA(\rho)\|\cA(\sigma))
\end{align}
which is a simple consequence of data-processing. 

We have defined \Renyi differential privacy based on the general \Renyi relative entropy $\DD_\alpha$ and while at this point any choice of a particular \Renyi relative entropy would seem justified, we will present some brief arguments that might hint that the sandwiched \Renyi relative entropy $\tilde D_\alpha$ should be the quantity of choice. First, $\tilde D_\alpha$ obeys data-processing for $\alpha\geq\frac12$, which makes it a valid choice for this whole range of $\alpha$, and it is equal to $D_{\max}$ in the limit $\alpha \rightarrow\infty$. Furthermore, $\tilde D_\alpha$ is the minimal \Renyi relative entropy, which means that 
\begin{align}
    \DD_\alpha(\rho\|\sigma) \geq \tilde D_\alpha(\rho\|\sigma)\,,
\end{align}
see e.g.~\cite{tomamichel2015quantum}. This implies that choosing $\tilde D_\alpha$ is the least restrictive choice and for a fixed $\alpha$ this \Renyi differential privacy would be implied by any other choice. Lastly, while we have seen that the measured \Renyi relative entropy has some undesirable properties, the sandwiched \Renyi relative entropy equals the regularized measured \Renyi relative entropy~\cite{tomamichel2015quantum}, 
\begin{align}\label{Eqn:sandwichedRegularized}
    \tilde D_\alpha(\rho\|\sigma) = \lim_{n\rightarrow\infty} \frac1n D_{M,\alpha}(\rho^{\otimes n}\|\sigma^{\otimes n}), 
\end{align}
giving it a close resemblance to the classical \Renyi differential privacy.

\subsection{Hypothesis testing}

At its core, differential privacy is a requirement on the probabilities associated to determining a used input based on the output of some information processing.  To gain a better intuition of the implications of the imposed restrictions, classical differential privacy can as often be reinterpreted in terms of hypothesis testing. This was first considered in~\cite{wasserman2010statistical} for $\epsilon$-differential privacy and then extended to $(\epsilon,\delta)$-differrential privacy in~\cite{kairouz2015composition}. Here we will present and discuss a quantum generalization of this analogy. Besides the intuitive formulation, we will see that it allows for a convenient graphical representation of differential privacy and simple proofs of some additional properties. Before we start, we remark that also \Renyi differential privacy has recently been discussed in terms of hypothesis testing~\cite{balle2020hypothesis} but we leave its quantum generalization for future work. 

The basic setup we will discuss is binary hypothesis testing between a state $\cA(\rho)$, the null hypothesis, and a state $\cA(\sigma)$, the alternative hypothesis. If we were able to discriminate between the two states, we could infer which input state was used. We are therefore interested in the corresponding probabilities of error, which are the Type I error $\alpha=\tr(\identity-M)\cA(\rho)$ of falsely rejecting the null hypothesis and the Type II error $\beta=\tr M\cA(\sigma)$ of falsely accepting it.  

As differential privacy has to hold for all neighbouring input states and all measurements we get the following set of restrictions on the Type I end Type II errors, 
\begin{align}
    1 - \alpha &\leq e^\epsilon \beta + \delta, \label{Eq:HT-bound-1}\\
    1 - \beta &\leq e^\epsilon \alpha + \delta \\
    \beta &\leq e^\epsilon (1-\alpha) + \delta \\
    \alpha &\leq e^\epsilon (1-\beta) + \delta, \label{Eq:HT-bound-4}
\end{align}
which follow by exchanging $\rho\leftrightarrow\sigma$ and $M\leftrightarrow(\identity-M)$ in the definition of differential privacy. Based on these inequalities we can define the privacy region of $(\epsilon,\delta)$-differential privacy as
\begin{align}
    \cR(\epsilon,\delta) = \{ (\alpha, \beta) \;|\; \text{Equations~\eqref{Eq:HT-bound-1}-\eqref{Eq:HT-bound-4} hold} \}. 
\end{align}
Next we define the privacy region of a quantum algorithm $\cA$ as 
\begin{align}
    &\cR(\cA) = \nonumber\\
    &\{ \left(\tr(\identity-M)\cA(\rho), \tr M\cA(\sigma)\right)\;|\; 0\leq M \leq\identity \;\text{and}\; \rho\sim\sigma \}. 
\end{align}
This allows us to state differential privacy in terms of privacy regions. 
\begin{theorem}
A quantum channel $\cA$ is $(\epsilon,\delta)$-differentially private if and only if
\begin{align}
    \cR(\cA) \subseteq \cR(\epsilon,\delta). 
\end{align}
\end{theorem}
\begin{proof}
The proof is a direct consequence of the above definitions. Note that we could have equivalently defined $\cR(\epsilon,\delta)$ with any subset of the Equations~\eqref{Eq:HT-bound-1}-\eqref{Eq:HT-bound-4}, for example only picking the first one, however this representation will be beneficial later on. 
\end{proof}
We continue by stating some properties of the risk regions. 
\begin{lemma}
The following holds.
\begin{itemize}
    \item{(Concatenation)} For arbitrary quantum algorithms $\cA$ and $\cN$ we have
    \begin{align}
        \cR(\cN\circ\cA) \subseteq \cR(\cA). 
    \end{align}
    \item{(Symmetry)} It holds that $\cR(\epsilon,\delta)$ is symmetric with respect to the line $\alpha+\beta=1$. 
\end{itemize}
\end{lemma}
\begin{proof}
The first statement follows by noting that 
\begin{align}
    \cR(\cN\circ\cA) = &\{ \left(\tr(\identity-N)\cA(\rho), \tr N\cA(\sigma)\right)  \\
    &\;|\; N=\cN^\dagger(M) \;\text{and}\; 0\leq M \leq\identity \;\text{and}\; \rho\sim\sigma \} \nonumber\\
    \subseteq &\cR(\cA),
\end{align}
which follows because $\cN^\dagger$ is completely positive and unital. The second statement follows directly from examining Equations~\eqref{Eq:HT-bound-1}-\eqref{Eq:HT-bound-4}. 
\end{proof}

\begin{figure*}
    \centering
    \begin{minipage}{0.47\textwidth}
        \centering
        \includegraphics[width=0.9\textwidth]{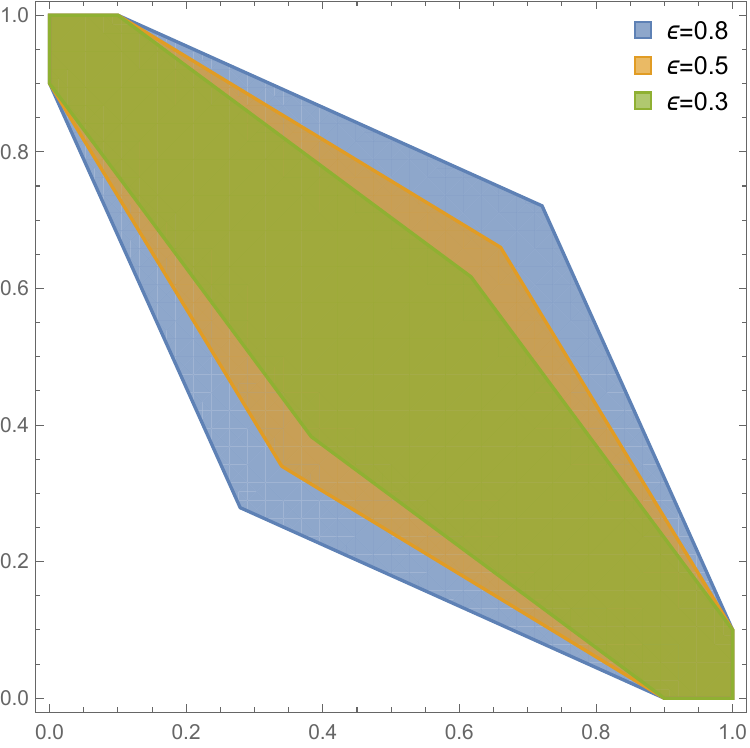} %
    \end{minipage}\hfill
    \begin{minipage}{0.47\textwidth}
        \centering
        \includegraphics[width=0.9\textwidth]{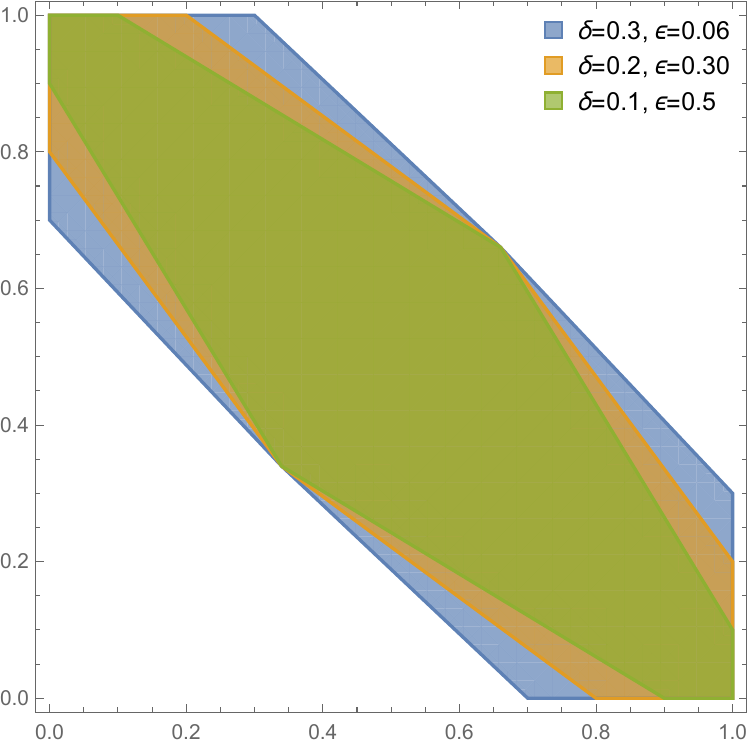} %
    \end{minipage}
\caption{\label{Fig:R-plots} Examples for $\cR(\epsilon,\delta)$. Left: plotted for $\delta=0.1$ and different values of $\epsilon$. Right: plotted for different values of $\delta$ and $\epsilon$, with $\epsilon$ chosen according to Lemma~\ref{Lem:HT-de-to}.}
\end{figure*}

Let us have a look at the graphical representation implied by the above definitions. In Figure~\ref{Fig:R-plots} we give several examples of $\cR(\epsilon,\delta)$ that illustrate the risk region of differential privacy. We also want to give a concrete numerical example of how the risk region of a channel contracts. Let's consider a simple example where only two qubit input states $\rho$ and $\sigma$ are available which are considered neighbouring. The states are chosen such that $E_1(\rho,\sigma)\leq \frac13$. For simplicity we consider a trivial algorithm to which we now want to add depolarizing noise such that the outputs become $(0.2,0.01)$-differentially private. From Equation~\eqref{Eq:Dp-Eg-bound} we can estimate that this should be the case if we choose $p\approx0.72$. To verify our observation numerically we simulate $\cR(\cD_p)$ by drawing random POVMs and compare the resulting pairs $(\alpha,\beta)$ to the desired privacy region. We can observe in Figure~\ref{Fig:R-Dp-plots} that this is indeed consistent with what we expected, namely that for $p=0.72$ all drawn points are within $\cR(0.2,0.01)$ while for smaller values of $p$ the noise is clearly not sufficient for $(0.2,0.01)$-differential privacy. 

Finally, we will see that phrasing differential privacy in terms of hypothesis testing does not only have advantages in terms of intuition, but also allows us to easily prove some useful results. 
\begin{lemma}\label{Lem:HT-de-to}
If $\cA$ is $(\epsilon,\delta)$-differentially private, then it is also $(\tilde\epsilon,\tilde\delta)$-differrentially private with $\tilde\delta\geq\delta$ and
\begin{align}
    e^{\tilde\epsilon} \geq \frac{(1-\tilde\delta)}{(1-\delta)}(1+e^\epsilon)-1. 
\end{align}
\end{lemma}
\begin{proof}
This follows easily from the graphical representation. We want to prove that 
\begin{align*}
    \cR(\epsilon,\delta) \subseteq \cR(\tilde\epsilon,\tilde\delta). 
\end{align*}
Let us consider the lower bounds on the risk region. It can easily be checked that they coincide at the point
\begin{align*}
    (\alpha^*,\beta^*) = (\frac{1-\delta}{1+e^\epsilon}, \frac{1-\delta}{1+e^\epsilon}), 
\end{align*}
which gives a corner point of the region. 
Since we require $\tilde\delta\geq\delta$, it suffices if
\begin{align*}
    \frac{1-\tilde\delta}{1+e^{\tilde\epsilon}} \leq \frac{1-\delta}{1+e^\epsilon}.
\end{align*}
This gives the claimed bound on $\tilde\epsilon$. 
\end{proof}
This result allows us to observe a certain trade-off between $\epsilon$ and $\delta$, in particular, by raising $\delta$ one can get differential privacy with a better value of $\epsilon$. This is also illustrated in the right part of Figure~\ref{Fig:R-plots}. 

\begin{figure*}
    \centering
    \begin{minipage}{0.33\textwidth}
        \centering
        \includegraphics[width=0.99\textwidth]{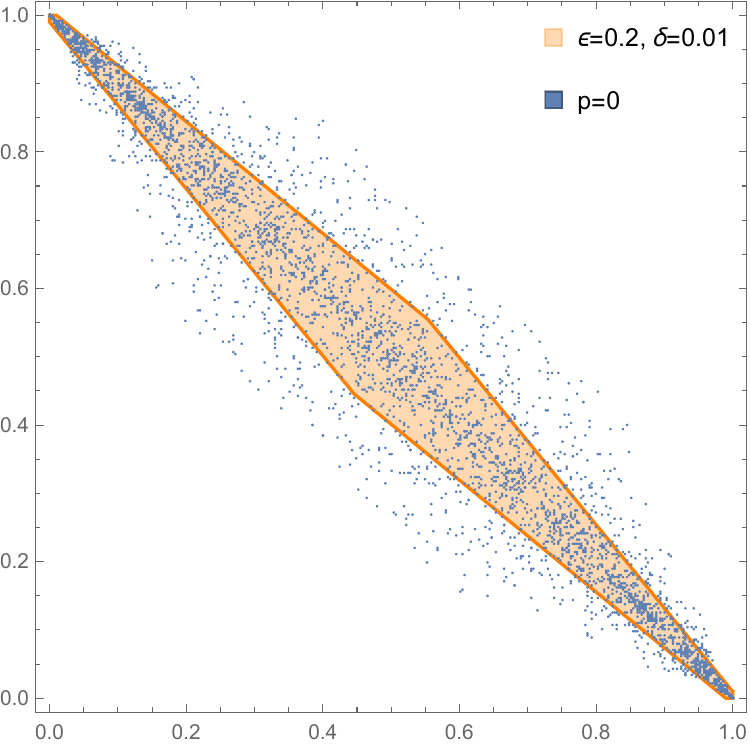} %
    \end{minipage}\hfill
    \begin{minipage}{0.33\textwidth}
        \centering
        \includegraphics[width=0.99\textwidth]{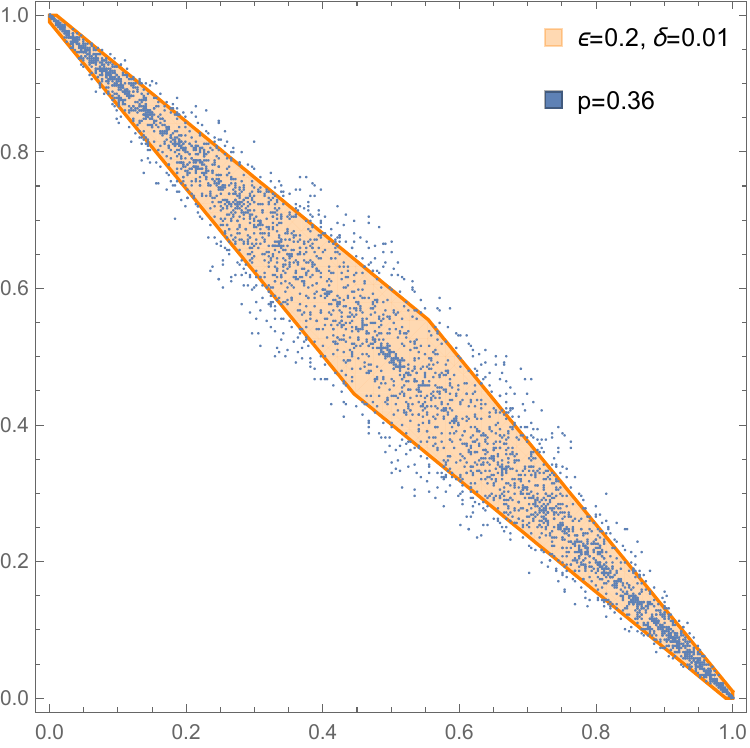} %
    \end{minipage}\hfill
    \begin{minipage}{0.33\textwidth}
        \centering
        \includegraphics[width=0.99\textwidth]{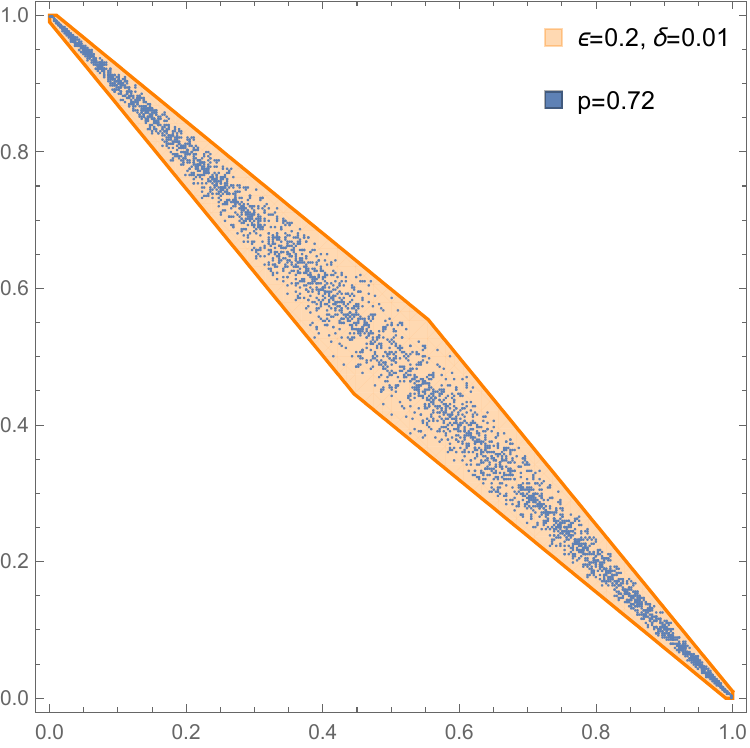} 
    \end{minipage}
\caption{\label{Fig:R-Dp-plots} Numerical values for points in $\cR(\cD_p)$ for different values of $p$. Points are based on two qubit states with $E_1(\rho,\sigma)\leq \frac13$ and 1000 randomly drawn POVMs. The orange region in the background corresponds to $\cR(0.2,0.01)$. }
\end{figure*}

\section{Conclusions}

In this work we gave a new approach to exploring quantum differential privacy via information theoretic tools. In particular, we used the quantum hockey-stick divergence to give a simple framework in which we can bound differential privacy parameters for practically relevant noise models such as quantum circuits and quantum neural networks implemented on near-term quantum devices. This includes comparing local and global depolarizing noise and contrasting achieving differential privacy with an undesirable decay in trace distance. 
On the way we showed several new properties of the said divergence and gave it a new operational interpretation. 

Given that our approach promises to be simpler and more powerful than the previously used ones we expect it to play a crucial role going forward when investigating differential privacy. 

Naturally we are left with some open problems. In Lemma~\ref{Lem:meas-rel-ent} we showed a bound on the measured relative entropy of the outputs of a differentially private algorithm. Classically several such results are known, including bounds on other entropic quantities such as the mutual information. An interesting open problem going forward is to find bounds on fully quantum quantities such as the quantum relative entropy or quantum mutual information. This would allow investigating connections to other privacy related quantities such as the quantum privacy funnel~\cite{datta2019convexity} generalizing work from the classical setting~\cite{Salamatin2020privacy}.

\section*{Acknowledgments} 
We would like to thank Theshani Nuradha, Ziv Goldfeld, Mark Wilde and an anonymous referee for spotting a mistake in a previous version of the work. 
This project has received funding from the European Union's Horizon 2020 research and innovation programme under the Marie Sklodowska-Curie Grant Agreement No. H2020-MSCA-IF-2020-101025848.  
DSF acknowledges financial support from the VILLUM FONDEN via the QMATH Centre of Excellence (Grant no. 10059)  the QuantERA ERA-NET Cofund in Quantum Technologies implemented within the European Union’s Horizon 2020 Programme (QuantAlgo project) via the Innovation Fund Denmark and from the European Research Council (grant agreement no. 81876). CR is partially supported by a Junior Researcher START Fellowship from the MCQST.  CR acknowledges funding by the Deutsche Forschungsgemeinschaft (DFG, German Research Foundation) under Germany's Excellence Strategy EXC-2111 390814868.

\bibliographystyle{abbrv}
\bibliography{library}

\appendix
\section{Some Lemmas} \label{Sec:Lemmas}
The following is a generalization of the Fuchs-van-de-Graaf inequality to general positive semi-definite operators proven in~\cite[Supplementary~Lemma~3]{CKW14}, see also~\cite[Appendix B]{wilde2020amortized} for an alternative proof. 
\begin{lemma}[\cite{CKW14}]\label{lem:FvdG-PSD}
For positive semi-definite, trace class operators $A$ and $B$ acting on a separable Hilbert space, we have that
\begin{equation}
\left\Vert A-B\right\Vert _{1}^{2}+4\left\Vert A^{1/2}B^{1/2}\right\Vert_{1}^{2}\leq\big(\operatorname{Tr}[A+B]\big)^2.
\end{equation}
\end{lemma}

Recently it was proven in~\cite{fawzi2022space} that for any quantum channel $\cT$ with $d$ its input and output dimension,
\begin{align}
    \eta_1(\cT) \leq \sqrt{1-\frac{\lambda_{\min}(\cT^\dagger\circ\cT)}{d^2}}, 
\end{align}
where $\lambda_{\min}(\cT^\dagger\circ\cT)=\min_{\Psi,\Phi} \langle\Psi|(\cT^\dagger\circ\cT)(|\Phi\rangle\langle\Phi|)|\Psi\rangle$. As we would like to generalize this result to the hockey-stick divergence, we extract the following lemma from their proof. 
\begin{lemma}[\cite{fawzi2022space}]\label{Lem:Fid-eigenvalue}
For any quantum channel $\cT$ with $d$ its input and output dimension and pure input states $\Psi$ and $\Phi$, 
\begin{align}
    F(\cT(\Psi),\cT(\Phi)) \geq \frac{\lambda_{\min}(\cT^\dagger\circ\cT)}{d^2}.
\end{align}
\end{lemma}
It was furthermore noted that for $\cT=\cN^{\otimes k}$ with $\cN$ a qubit channel, one has 
\begin{align}
    F(\cT(\Psi),\cT(\Phi)) \geq \left(\frac{\lambda_{\min}(C_{\cN^\dagger\circ\cN})}{4}\right)^k,
\end{align}
where $C_{\cN^\dagger\circ\cN}$ is the Choi matrix of $\cN^\dagger\circ\cN$, and if $\cN$ is also non-unital one gets 
\begin{align}
    \lambda_{\min}(C_{\cN^\dagger\circ\cN})>0.
\end{align}
We also need bounds on $D_{\max}^\epsilon(\rho\|\sigma)$. As remarked earlier, our definition of $D_{\max}^\epsilon$ is a bit different than the one usually used in the quantum information literature, as it uses a different distance measure. However, to apply a known result we need to compare our definition to the usual one. The standard smooth $\max$-relative entropy, $D_{\max,s}$, is defined as
\begin{align}
D_{\max,s}^\epsilon(\rho\|\sigma) = \inf_{\bar\rho\in B_s^\epsilon(\rho)} D_{\max}(\bar\rho\|\sigma)
\end{align}
and 
\begin{align}
B_s^\epsilon(\rho) = \{\bar\rho : \bar\rho\in\cSs{} \wedge P(\rho,\bar\rho)\leq \epsilon \},
\end{align}
where $P$ is the purified distance, i.e. the minimal trace distance between purifications of the states. See e.g.~\cite[Definition 3.15]{tomamichel2015quantum} for a discussion of this quantity. 
We prove the following auxiliary lemma.
\begin{lemma}\label{Lem:aux-Dmax}
Let $\rho,\sigma\in\cSe{}$, then
\begin{align}
   D_{\max}^\epsilon(\rho\|\sigma) \leq D_{\max,s}^\epsilon(\rho\|\sigma) 
\end{align}
\end{lemma}
\begin{proof}
The claim follows immediately by showing that $B_s^\epsilon(\rho) \subseteq B^\epsilon(\rho)$. To that end observe, 
\begin{align}
    B_s^\epsilon(\rho) &= \{\bar\rho : \bar\rho\in\cSs{} \wedge P(\rho,\bar\rho)\leq \epsilon \} \\
    &\subseteq \{\bar\rho : \bar\rho\in\cSs{} \wedge E_1(\rho,\bar\rho)\leq \epsilon \} \\
    &\subseteq \{\bar\rho : \bar\rho\in\cPo{} \wedge E_1(\rho,\bar\rho)\leq \epsilon \} \\
    &= B^\epsilon(\rho),
\end{align}
where the second inclusion is clear, but the first needs some justification. Note that 
\begin{align}
    E_1(\rho,\bar\rho) &\leq \max\{ E_1(\rho,\bar\rho) , E_1(\bar\rho,\rho) \} \\
    &= \Delta(\rho,\bar\rho) \\
    &\leq P(\rho,\bar\rho), 
\end{align}
where $\Delta$ is the generalized trace distance and the equality holds by its definition. The first inequality is immediate and the second is a generalized Fuchs-van-de-Graaf type inequality, see e.g.~\cite[Lemma 3.17]{tomamichel2015quantum}. This concludes the proof.  
\end{proof}
This enables us to use the following result.
\begin{lemma}
Let $0\leq\epsilon\leq 1$ and $\alpha\in(1,\infty)$, then 
\begin{align}
    D^\epsilon_{\max}(\rho\|\sigma) \leq \DD_\alpha(\rho\|\sigma) + \frac{g(\epsilon)}{\alpha-1},
\end{align}
where $g(\epsilon)=-\log(1-\sqrt{1-\epsilon^2})$ and $\DD_\alpha$ any quantum \Renyi divergence. 
\end{lemma}
\begin{proof}
The result immediately follows from~\cite[Proposition 6.22]{tomamichel2015quantum}, which shows the same result for $D^\epsilon_{\max,s}(\rho\|\sigma)$, and the fact that $D^\epsilon_{\max}(\rho\|\sigma)\leq D^\epsilon_{\max,s}(\rho\|\sigma)$ proved in Lemma~\ref{Lem:aux-Dmax}.
\end{proof}

\section{\Renyi relative entropy}\label{Sec:App-Renyi}

In Section~\ref{Sec:Renyi}, we introduced the new concept of \Renyi quantum differential privacy and based it on a very general framework of \Renyi relative entropies that solely have to fulfil certain properties as listed in~\cite{tomamichel2015quantum}. For completeness we list those properties here. 

A quantum \Renyi divergence is a quantity $\DD(\cdot \| \cdot)$ that fulfills the following properties: 
\begin{enumerate}
    \item \textbf{Continuity:} $\DD(\rho \| \sigma)$ is continuous in $\rho$ and $\sigma$, wherever $\rho\neq 0$ and $\sigma >> \rho$. 
    \item \textbf{Unitary invariance:} $\DD(\rho\|\sigma) = \DD(U\rho U^\dagger\| U\sigma U^\dagger)$ for any unitary $U$. 
    \item \textbf{Normalization:} $\DD(1\|\frac12) = \log2$.
    \item \textbf{Order:} If $\rho\geq\sigma$, then $\DD(\rho\|\sigma)\geq 0$ and if $\rho\leq\sigma$ then $\DD(\rho\|\sigma)\leq 0$. 
    \item \textbf{Additivity:} $\DD(\rho\otimes\tau\|\sigma\otimes\omega) = \DD(\rho\|\sigma) + D(\tau\|\omega)$. 
    \item \textbf{General mean:} There exists a continuous and strictly monotonic function $g$ such that $Q:=g(\DD)$ satisfies,
    \begin{align*}
        &Q(\rho\oplus\tau\|\sigma\oplus\omega) \\
        &= \frac{\tr(\rho)}{\tr(\rho+\tau)} Q(\rho\|\sigma) + \frac{\tr(\tau)}{\tr(\rho+\tau)} Q(\tau\|\omega).
    \end{align*}
    \item \textbf{Positive Definiteness:} $\DD(\rho\|\sigma)\geq 0$ with equality iff $\rho=\sigma$. 
    \item \textbf{Data-processing:} $\DD(\rho\|\sigma)\geq \DD(\cN(\rho)\|\cN(\sigma)).$
    \item Either \textbf{joint convexity} or \textbf{joint concavity} of $Q$. 
    \item \textbf{Dominance:} For $\sigma\leq\sigma'$, one has $\DD(\rho\|\sigma)\geq\DD(\rho\|\sigma')$. 
\end{enumerate}
In the classical case, properties 1-6 uniquely define the \Renyi relative entropies. This is not the case in the quantum setting where one additionally requires the operationally motivated properties 7-10. 

Finally, a family of quantum \Renyi relative entropies is a one-parameter family $\alpha\rightarrow\DD_\alpha(\cdot\|\cdot)$ of quantum \Renyi relative entropies such that for some open interval containing $1$, the family is monotonically increasing in $\alpha$. 

\end{document}